\documentclass[a4paper,12pt,english]{article}
\usepackage{graphicx,url,amsfonts,amsmath,latexsym,amssymb,amsthm,fixmath}
\usepackage{color,geometry}

\usepackage[T1]{fontenc}
\usepackage{lmodern}
\usepackage{babel}
\usepackage[babel]{microtype}

\usepackage{hyperref}

\graphicspath{{figures/}}
\def\svgpath{figures/}

\newcommand{\includesvg}[2][]{%
\def\tempa{#1}\def\tempb{}%
\ifx\tempa\tempb\else\let\svgwidth\tempa\fi
\input{\svgpath#2_generated.pdf_tex}
}

\def\ie/{\emph{i.e.},}
\def\CW/{crossing weight}
\def\MCW/{minimal \CW/}
\def\dash/{\hbox{-}\nobreak\hskip0pt\relax}
\def\short/{short}

\newcommand{\define}[1]{\textbf{\textsl{#1}}}
\newcommand{\name}[1]{\hbox{#1}}
\let\phi\varphi
\let\epsilon\varepsilon

\newcommand{\inv}{\ensuremath{^{-1}}}
\newcommand{\cut}{\mathbin{\setminus\mskip-6.5mu\setminus}}

\newcommand{\surf}{{\mathcal S}}
\newcommand{\unicov}[1][\surf]{\tilde{#1}}
\newcommand{\lift}[1]{\tilde{#1}}

\newcommand{\uniH}{\tilde{H}}

\newcommand{\unic}{\tilde{c}}

\newcommand{\unie}{\tilde{e}}
\newcommand{\unih}{\tilde{h}}
\newcommand{\mcO}{O}
\newcommand{\Z}{\mathbb Z}
\newcommand{\R}{\mathbb{R}}

\newcommand{\region}[1]{\Pi_{#1}}

\newcommand{\cycov}[1]{\surf_{#1}}
\newcommand{\cyproj}{\varphi_c}
\newcommand{\cyH}{H_c}

\newcommand{\trans}{\tau}

\newcommand{\NEW}{\mathrm{new}}
\newcommand{\TYPE}{\mathrm{type}}
\newcommand{\JOIN}{\mathrm{join}}
\newcommand{\NEXT}{\mathrm{next}}
\newcommand{\NONE}{\mathrm{NONE}}
\newcommand{\ID}{\mathrm{id}}
\newcommand{\len}[1]{\left\lvert#1\right\rvert}
\newcommand{\ttstyle}[1]{\hbox{\texttt{#1}}}

\newtheorem{theorem}{Theorem}
\newtheorem{lemma}[theorem]{Lemma}
\newtheorem{proposition}[theorem]{Proposition}
\newtheorem{cor}[theorem]{Corollary}

\title{On the homotopy test on surfaces}
\author{Francis Lazarus\thanks{%
    GIPSA-Lab, CNRS, Grenoble, France;
    \url{Francis.Lazarus@grenoble-inp.fr}}%
  \and Julien Rivaud \thanks{%
    GIPSA-Lab, CNRS, Grenoble, France;
    \url{Julien.Rivaud@grenoble-inp.fr}}%
}

\begin{document}
\maketitle
\frenchspacing

\begin{abstract}
  Let~$G$~be a graph cellularly embedded in a surface~$\surf$.
  Given two closed walks $c$~and~$d$ in~$G$, we take
  advantage of the RAM~model to describe
  linear time algorithms to decide if $c$~and~$d$ are homotopic
  in~$\surf$, either freely or with fixed basepoint. We~restrict
  $\surf$ to be orientable for the free homotopy test, but allow
  non-orientable surfaces when the basepoint is fixed. After
  $O(|G|)$~time preprocessing independent of $c$~and~$d$, our
  algorithms answer the homotopy test in $O(|c|+|d|)$~time, where
  $|G|$,~$|c|$ and~$|d|$ are the respective numbers of edges of
  $G$,~$c$ and~$d$. As a byproduct we obtain linear
  time algorithms for the word problem and the conjugacy problem in
  surface groups. These results were previously announced by Dey and
  Guha (1999). Their approach was based on small cancellation theory
  from combinatorial group theory. However, several flaws in their
  algorithms make their approach fails, leaving the
  complexity of the homotopy test problem still open. We~present a
  geometric approach, based on previous works by Colin~de~Verdi\`ere
  and Erickson, that provides optimal homotopy tests.
\end{abstract}

\section{Introduction}
\label{sec:intro}
Computational topology of surfaces has received much attention in the
last two decades. Among the notable results we may mention the test of
homotopy between two cycles on a surface~\cite{dg-tcs-99}, the
computation of a shortest cycle homotopic to a given
cycle~\cite{ce-tnspc-10}, or the computation of optimal homotopy and
homology bases~\cite{ew-gohhg-05}. In their 1999 paper, Dey and Guha
announced a linear time algorithm for testing whether two curves on a
triangulated surface are freely homotopic. This appeared as a major
breakthrough for one of the most basic problem in computational
topology. Dey and Guha's approach relies on results by
Greendlinger~\cite{g-dacwp-60} for the conjugacy problem in one
relator groups satisfying some small cancellation condition. In the
appendix, we show several subtle flaws in the paper of Dey and
Guha~\cite{dg-tcs-99} that invalidate their approach and leave little
hope for repair. Inspired by the recent work of Colin~de~Verdi\`ere
and Erickson~\cite{ce-tnspc-10} for computing a shortest cycle in a
free homotopy class, we propose a different geometric approach and
confirm the results of Dey and Guha for orientable surfaces. As
commonly assumed in computational topology, we shall analyse the
complexity of our algorithms with the uniform cost RAM~model of
computation~\cite{ahu-daca-74}. A notable feature of this model is the
ability to manipulate arbitrary integers in constant time per
operation and to access an arbitrary memory register in constant time.

In a first part we consider the homotopy test for curves with fixed
endpoints drawn in a graph cellularly embedded in a
surface~$\surf$. This test reduces to decide if a loop is contractible
in~$\surf$, \ie/ null-homotopic, since a curve~$c$ is homotopic to
a curve~$d$ with fixed endpoints if and only if the
concatenation~$c\cdot d\inv$ is contractible. The contractibility test
was already considered by Dey and Schipper~\cite{ds-ntcps-95} using a
partial and implicit construction of the universal cover
of~$\surf$. Indeed, a curve is null-homotopic in~$\surf$ if and only
if its lift is closed in the universal cover of~$\surf$. Given a
closed curve~$c$, Dey and Schipper detect if $c$~is
null-homotopic in $O(|c|\log g)$~time, where $g$~is the genus
of~$\surf$. Their implicit construction is relatively complex and does
not seem to extend to handle the free homotopy test. Our solution to
the contractibility test also relies on a partial construction of the
universal cover. We use the more explicit construction of Colin~de~Verdi\`ere
and Erickson~\cite[Sec.  3.3~and~4]{ce-tnspc-10} for
tightening paths. It amounts to build a convex region of the universal
cover (with respect to some hyperbolic metric) large enough to contain
a lift of~$c$. An argument \emph{\`{a} la Dehn} shows that this region
can be chosen to have size~$O(|c|)$, leading to our first theorem:
\begin{theorem}[Contractibility test]\label{th:contractibility-test}
  Let~$G$ be a graph of complexity~$n$ cellularly embedded in a
  surface~$\surf$, not necessarily orientable. We can preprocess~$G$
  in $O(n)$~time, so that for any loop~$c$ on~$\surf$ represented as a
  closed walk of $k$~edges in~$G$, we can decide whether $c$~is
  contractible or not in $O(k)$~time.
\end{theorem}

We next study the free homotopy test, that is deciding if two cycles
$c$~and~$d$ drawn in a graph~$G$ cellularly embedded in~$\surf$ can be
continuously deformed one to the other. By
theorem~\ref{th:contractibility-test}, we may assume that none of $c$~and~$d$
is contractible. Our strategy is the following. We first build (part
of) the cyclic covering of~$\surf$ induced by the cyclic subgroup
generated by $c$ in the fundamental group of~$\surf$. We denote by~$\cycov{c}$
this covering. Assuming that $\surf$~is orientable, $\cycov{c}$ is~a topological
cylinder\footnote{If $\surf$~is non
  orientable, its cyclic coverings can be either cylinders or M\"obius
  rings.}, and we call any of its non-contractible simple cycles a
\emph{generator}. Since the generators of~$\cycov{c}$ are freely
homotopic, their projection on~$\surf$ are freely homotopic
to~$c$. Our next task is to extract from~$\cycov{c}$ a
\emph{canonical} generator~$\gamma_R$ whose definition only depends on
the isomorphism class of~$\cycov{c}$. To this end, we lift in~$\cycov{c}$
the graph~$G$ of~$\surf$ and we endow~$\cycov{c}$ with the
corresponding cross-metric introduced by Colin~de~Verdi\`ere and
Erickson~\cite{ce-tnspc-10}. The set of generators that are minimal
for this metric form a compact annulus in~$\cycov{c}$. We
eventually define $\gamma_R$~as the ``right'' boundary of this
annulus. We perform the same operations starting with~$d$ instead of~$c$
to extract a canonical generator~$\delta_R$ of~$\cycov{d}$. From
standard results on covering spaces~\cite[\textsection
V.6]{m-bcat-91}, we know that $\cycov{c}$~and~$\cycov{d}$ are
isomorphic covering spaces if $c$~and~$d$ are freely homotopic. It
follows that $c$~and~$d$ are freely homotopic if and only if
$\gamma_R$~and~$\delta_R$ have equal projections on~$\surf$.
Proving that $\gamma_R$~and~$\delta_R$ can be constructed in time proportional
to $|c|$~and~$|d|$ respectively, we finally obtain:
\begin{theorem}[Free homotopy test]\label{th:free-homotopy-test}
  Let~$G$ be a graph of complexity~$n$ cellularly embedded in an
  orientable surface~$\surf$. We can preprocess~$G$
  in $O(n)$~time, so that for any cycles $c$~and~$d$ 
  on~$\surf$ represented as closed walks with a total number of~$k$
  edges in~$G$, we can decide if $c$~and~$d$
  are freely homotopic in $O(k)$~time.
\end{theorem}

The \emph{word problem} in a group presented by generators and
relators is to decide if a product of generators and their inverses,
called a \emph{word}, is the unit in the group. The \emph{conjugacy problem}
is to decide if two words represent conjugate elements in the group.  The
length of a word is its number of factors. As an
immediate consequence of our two theorems, we can solve the word
problem and the conjugacy problem in surface groups in optimal linear
time. More precisely, suppose we are given a presentation by
generators and by a single relator of the fundamental group of a
compact surface~$\surf$ of genus~$g$ without boundary. After $O(g)$~time
preprocessing, we can solve the word problem in time proportional to
the length of the word. Moreover, if $\surf$~is orientable we can
report if two words are conjugate in time proportional to
their total length. The preprocessing reduces to build
a cellular embedding of a wedge of loops in~$\surf$, with one loop for
each generator; the rotation system of this embedding (see the
Background Section for a definition) is easily
deduced from the relator
of the group. Any word can then be interpreted as a walk in this cellular
embedding so that we can directly apply the previous theorems.
The word and conjugacy problems have a long standing history starting with
Dehn's seminal papers~\cite{s-pgtt-87}. Recent developments include
linear time solutions to the word problem in much larger classes of
groups comprising hyperbolic
groups~\cite{da-cdawp-85,h-whghr-00,hr-swprt-01}. We emphasize that
such developments assume a multi-tape Turing machine as a model of
computation and, most importantly, that the size of the group
presentation is considered as a constant. In our case, the group
itself is part of the input and, after the preprocessing phase, the
decision problems have linear time solutions \emph{independent of the
  genus of the surface}.

\paragraph{Organization of the paper.}
We start recalling some necessary terminology and properties of
surfaces, coverings and cellular embeddings of graphs in
Section~\ref{sec:background}. We solve the contractibility test and
prove Theorem~\ref{th:contractibility-test} in
Section~\ref{sec:contractibility-test}.  The proof of
Theorem~\ref{th:free-homotopy-test} for the free homotopy test is
given in Section~\ref{sec:free-homotopy-test}. We eventually give some
counter-examples to the results of Dey and Guha~\cite{dg-tcs-99} in
the appendix.

\section{Background}
\label{sec:background}
We review some basic definitions and properties of surfaces and their covering
spaces, as well as combinatorial embeddings of graphs. We refer the reader to
Massey~\cite{m-bcat-91} or Stillwell~\cite{s-ctcgt-93} for further details on
covering spaces.

\paragraph{Surfaces.}
We only consider surfaces without boundary. A \define{surface} (or
$2$\dash/manifold)~$\surf$ is a connected, Hausdorff topological space where
each point has a neighborhood homeomorphic to the plane.
A compact surface is homeomorphic to a sphere where either:
\begin{itemize}
    \item $g\geq0$~open disks are removed and a handle (\ie/ a
      perforated torus
      with one boundary component) is attached to each
        resulting circle, or
    \item $g\geq1$~open disks are removed and a M\"obius band is attached to
        each resulting circle.
\end{itemize}
The surface is called \define{orientable} in the former case and
\define{non-orientable} in the latter case. In both cases, $g$~is the
\define{genus} of the surface.

A \define{path} in a surface~$\surf$, is a continuous map $p: [0,1]\to \surf$. A
\define{loop} is a path~$p$ whose endpoints $p(0)$~and~$p(1)$ coincide. This
common endpoint is called the \define{basepoint} of the loop.

\paragraph{Homotopy and fundamental group.}
Two paths $p$,~$q$ in~$\surf$ are \define{homotopic} (with fixed
endpoints) if there is a continuous map $h : [0, 1]\times [0, 1]
\rightarrow \surf$ such that $h(0, t) = p(t)$ and $h(1, t) = q(t)$ for
all~$t$, and $h(\cdot, 0)$~and~$h(\cdot, 1)$ are constant maps. Being
homotopic is an equivalence relation. The set of homotopy classes of
loops with given basepoint~$x\in\surf$ forms a group where the
operation in the group corresponds to the concatenation of the
loops. This group is called the \define{fundamental group} of~$\surf$
and denoted by~$\pi_1(\surf,x)$. The homotopy class of a loop~$c$ is
denoted by~$[c]$. The loop~$c$ is said \define{contractible}, or
null-homotopic, if $c$~is homotopic to the constant loop, \ie/ if
$[c]$ is the identity of~$\pi_1(\surf,x)$. The fundamental group of
the orientable surface~$\surf$ of genus~$g\geq 1$ admits finite
minimal presentations composed of $2g$~generators and one relator
expressed as the product of $4g$~generators and their inverses. A~group
defined by a set~$A$ of generators and a set~$R$ of relators is
denoted~$\langle A \mathrel; R\rangle$. In particular,
$\pi_1(\surf,x)$~is isomorphic to the \define{canonical} presentation
$\langle a_1, b_1, \dotsc,a_g,b_g \mathrel; a_1b_1a_1\inv b_1\inv
\dotsm a_gb_ga_g\inv b_g\inv\rangle$ but not all minimal presentations
of~$\pi_1(\surf,x)$ are canonical. Likewise, the fundamental group of
the non-orientable surface of genus~$g$ has a presentation with $g$~generators
and one relator of length~$2g$.

Two loops $c$,~$d$ in~$\surf$ are \define{freely homotopic} if there is a
continuous map $h:[0, 1]\times [0, 1] \rightarrow \surf$ such that $h(0,\cdot) =
c$, $h(1,\cdot) = d$ and $h(s,0) = h(s,1)$ for all~$s$. The loops $c$~and~$d$ of
respective basepoints $x$~and~$y$ are freely homotopic if and only if
$[c]$~and~$[u\cdot d \cdot u\inv]$ are conjugate in~$\pi_1(\surf,x)$ for any
path~$u$ linking $x$~to~$y$.

\paragraph{Covering spaces.}
A \define{covering space} of the surface~$\surf$ is a surface~$\surf'$ together
with a continuous surjective map~$\pi: \surf'\to\surf$ such that every~$x
\in\surf$ lies in an open neighborhood~$U$ such that $\pi\inv(U)$~is a disjoint
union of open sets in~$\surf'$, each of which is mapped homeomorphically
onto~$U$ by~$\pi$. The map~$\pi$ is called a \define{covering map}; it induces a
monomorphism $\pi_*: \pi_1(\surf',y)\to\pi_1(\surf,\pi(y))$, so that
$\pi_1(S',y)$~can~be considered as a subgroup of~$\pi_1(S,\pi(y))$. If $p$~is a
path in~$\surf$ and $y\in \surf'$ with $\pi(y)=p(0)$, then there exists a unique
path $q : [0,1]\to\surf'$, called a \define{lift} of~$p$, such that $\pi\circ q
= p$ and $q(0)=y$.

A \define{morphism} between the covering spaces
$(\surf',\pi)$~and~$(\surf'',\pi')$ of~$\surf$ is a continuous map
$\phi: \surf'\to\surf''$ such that $\pi'\circ\phi = \pi$. Up to
isomorphism, each surface~$\surf$ has a unique simply connected
covering space, called its \define{universal cover} and
denoted~$\unicov$. Unless $\surf$~is a sphere or a projective plane
$\unicov$~has the topology of a plane. More generally, for every
subgroup~$\mathcal G$ of~$\pi_1(\surf,x)$ there is a covering space of~$\surf$,
unique up to isomorphism, whose fundamental group is conjugate to~$\mathcal G$
in~$\pi_1(\surf,x)$. If~$\mathcal G$~is cyclic and generated by the homotopy
class of a non-contractible loop~$c$ in $\surf$, this covering space
is called the $c$\dash/\define{cyclic cover} and denoted
by~$\cycov{c}$. The $c$\dash/cyclic cover can be constructed as follows:
take a lift~$\unic$ of~$c$ in the universal cover~$\unicov$ and let
$\trans$~be~the unique automorphism of~$\unicov$ sending~$\unic(0)$
to~$\unic(1)$; then $\cycov{c}$~is the quotient of the action
of~$\left<\trans\right>$ on~$\unicov$.  When $\surf$~is orientable,
$\cycov{c}$~has the topology of a cylinder whose generators project
on~$\surf$ to loops that are freely homotopic to~$c$ or its
inverse. Conversely, every loop freely homotopic to~$c$ has a closed
lift generating the fundamental group of~$\cycov{c}$. In
Section~\ref{sec:free-homotopy-test} we shall use the
following properties of curves on cylinders. We assume that the
considered curves
are in general position: all (self-)intersections are transverse and
have multiplicity two, \ie/ exactly two curve pieces cross at an intersection.
\begin{lemma}\label{lem:cylinder-1}
  Let~$c$~be a loop obtained as the concatenation of $k$~simple paths
  on a cylinder. Then $c$~is freely homotopic to the $i$\dash/th power of
  a generator of the cylinder with~$|i| < k$.
\end{lemma}
\begin{proof}
  We view
  the cylinder as a punctured plane; so that the $i$\dash/th power of a
  generator has winding number~$i$ with respect to the puncture. The
  winding number of~$c$ is the sum of the angular extend of each of
  its subpaths divided by~$2\pi$. But the angular extend of a simple
  path has absolute value strictly smaller than~$2\pi$.
\end{proof}
\begin{lemma}\label{lem:cylinder-2}
  A self-intersecting generator of a cylinder has a contractible
  closed subpath.
\end{lemma}
\begin{proof}
  Consider a self-intersecting generator~$\gamma$ and define a bigon
  of~$\gamma$ as a disk bounded by two subpaths of~$\gamma$. Applying
  local homotopies to~$\gamma$ we can remove all of its bigons one by
  one. If~$\gamma$~is still self-intersecting, it must have a
  contractible subpath by \cite[Lemma 1.4]{hs-ics-85}. If~$\gamma$~is
  simple, we consider the loop just before we remove the last
  bigon. It is easily seen that this loop has a contractible
  subpath. In both cases we have found a contractible subpath that
  corresponds to a contractible subpath of the initial loop if we undo
  the local homotopies.

  For completeness, we also give a self-contained proof. We again view
  the cylinder as a punctured plane. Consider the domain of~$\gamma$
  as the unit circle~$\R/\Z$. A \emph{loop-segment} is a closed
  connected part $[x,y]\subset\R/\Z$ whose image by $\gamma$ is a
  simple loop, \ie/ such that $\gamma(x)=\gamma(y)$ and the
  restriction of~$\gamma$ to~$[x,y)$ is one-to-one. Let
  $[x_1,y_1],\ldots,[x_k,y_k]$ be a maximal set of pairwise interior
  disjoint loop-segments. By the general position assumption, $k$~is
  finite. We denote by~$\gamma_i$ the restriction of~$\gamma$ to~$[x_i,y_i]$.
  It is easily seen that $\R/\Z \setminus \cup_{i=1}^k
  [x_i,y_i]$ has at most $k$~connected components (exactly~$k$ if we take into
  account that each self-intersection has multiplicity two) whose
  images are simple paths composing a loop~$c$. If some loop~$\gamma_i$
  is contractible, then we are done. Otherwise, being simple, each~$\gamma_i$
  has winding number~$\varepsilon_i=\pm 1$. Let~$w(c)$~denote
  the winding number of~$c$. By additivity of the winding
  number, we have
  \[ \sum_{i=1}^{k} \varepsilon_i + w(c) = 1.
  \]
  By the preceding lemma we have $|w(c)|<k$, implying that
  $\varepsilon_j=1$ for some~$j\in[1,k]$. It ensues that the
  complementary part of~$\gamma_j$ in~$\gamma$ is a loop with winding
  number~$1-\varepsilon_j = 0$, hence contractible.
\end{proof}

\paragraph{Cellular embeddings of graphs.}
All the considered graphs may have loop edges and multiple edges. We denote by
$V(G)$~and~$E(G)$ the set of vertices and edges of a graph~$G$, respectively. A
graph~$G$ is \define{cellularly embedded} on a surface~$\surf$ if every open
face of (the embedding of)~$G$ on~$\surf$ is a disk. Following Mohar and
Thomassen~\cite{mt-gs-01}, the embedding of~$G$ can be encoded by adjoining to
the data of~$G$ a \define{rotation system}. The rotation system provides for
every vertex in~$V(G)$ a cyclic permutation of its incident edges. For a finite
graph~$G$, storing a cellular embedding takes a space linear in the
\define{complexity} of~$G$, that is, in its total number of vertices and edges.
A \define{facial walk} of~$G$ is then obtained by the face traversal procedure
described in~\cite[p. 93]{mt-gs-01}. Rotation systems can be implemented
efficiently~\cite{e-dgteg-03,l-gem-82} so that we can traverse the neighbors of
a vertex in time proportional to its degree or obtain a facial walk in time
proportional to its length.

Any graph~$G$ cellularly embedded in~$\surf$ has a \define{dual graph}
denoted~$G^*$ whose vertices and edges are in one-to-one
correspondence with the faces and edges of~$G$ respectively; if two
faces of~$G$ share an edge~$e \in E(G)$ its \define{dual edge}~$e^*
\in E(G^*)$ links the corresponding vertices of~$G^*$. This dual graph
can be cellularly embedded on~$\surf$ so that each face of~$G$
contains the matching vertex of~$G^*$ and each edge~$e^*$ dual
of~$e\in E(G)$ crosses only~$e$, only once.

Let~$(\surf',\pi)$ be a covering space of~$\surf$. The \define{lifted
  graph} $G'= \pi\inv(G)$ is cellularly embedded in~$\surf'$. The
restriction of~$\pi$ from the star of each vertex~$x\in V(G')$ to the
star of~$\pi(x)\in V(G)$ is an isomorphism. Note that $\surf'$~can~be
non compact even though $\surf$~is compact. In that case, $G'$~has an
infinite number of edges and vertices.

\paragraph{Regular paths and crossing weights.}
For the free homotopy test in Section~\ref{sec:free-homotopy-test}, we
make use of the cross metric surface model~\cite{ce-tnspc-10}. Let~$G$~be
a graph cellularly embedded in~$\surf$. A path $p$~in~$\surf$ is
\define{regular} for~$G$ if every intersection point of $p$~and~$G$ is
an endpoint of~$p$ or a transverse crossing, \ie/ has a neighbourhood
in which $p \cup G$~is homeomorphic to two perpendicular line segments
intersecting at their midpoint. The
\define{\CW/} with respect to~$G$ of a regular path~$p$ is the
number~$\len{p}$ of its transverse crossings and is always finite. In
particular, if $p$~is a path of the dual graph~$G^*$ then $p$~is
regular for~$G$ and $\len{p}$~is the number of edges in~$p$.

\paragraph{Homotopy encoding in cellular graph embeddings.}
Consider a graph~$G$ cellularly embedded in~$\surf$. If $H$~is a
subgraph of~$G$, we will denote by~$\surf\cut H$ the surface obtained
after cutting~$\surf$ along~$H$. If $\surf\cut H$~is a topological
disk, then $H$~is called a \define{cut graph}. A cut graph can be
computed in linear time~\cite{ccl-osaew-10,e-dgteg-03}. Note that a
cut graph defines a cellular embedding in~$\surf$ with a unique
face. Let~$T$ be a spanning tree of a cut graph~$H$ and consider the set
of edges $A := E(H)\setminus E(T)$. See
Figure~\ref{fig:cut-graph}.
\begin{figure}
  \centering
  \includesvg[.85\linewidth]{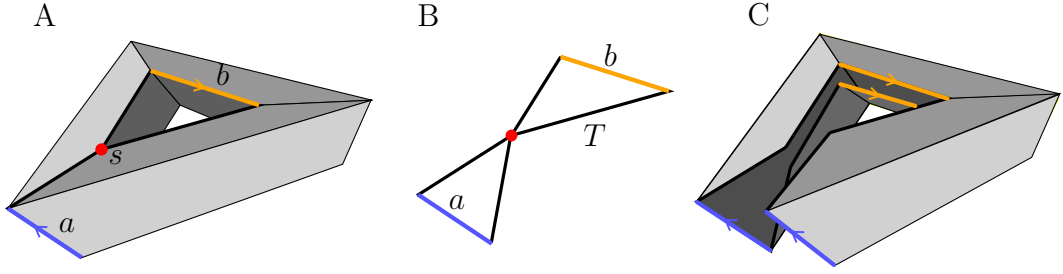}
  \caption{(A)~A quadrangulation of a torus ($g=1$).\allowbreak\ (B)~A cut graph
  of the torus with a spanning tree~$T$. The set $A=\{a,b\}$ of non-tree edges
  contains $2g = 2$ elements.\allowbreak\ (C)~The torus cut through the cut
  graph can be flattened into the plane as in
  Figure~\ref{fig:homotopy-encoding}. Here~$f_H|_A = aba\inv b\inv$.}
  \label{fig:cut-graph}
\end{figure}
If $\surf$~is a compact surface of
genus~$g$, Euler's formula easily implies that $A$~contains either $2g$~or~$g$
edges depending on whether $\surf$~is respectively orientable
or non-orientable. For each vertex~$s\in H$, we have $\pi_1(\surf,s)
\cong\, \langle A \mathrel; f_H|_A\rangle$, where $f_H$~is the facial
walk of the unique face of~$H$ and $f_H|_A$~denotes its restriction to
the edges in~$A$. Indeed, if we contract~$T$ to the vertex~$s$
in~$\surf$, the graph~$H$ becomes a bouquet of circles whose
complementary set in~$\surf$ is a disk bounded by the facial walk~$f_H|_A$.
The above group presentation then follows from the classical
Seifert and Van Kampen Theorem~\cite[Chap. IV]{m-bcat-91}.

Let~$c$ be
a closed walk in~$G$ with basepoint~$x\in V(H)$. Denote by~$T(s,x)$
the unique simple path in~$T$ from~$s$~to~$x$. We can express the
homotopy class~$[c']$ of the closed walk $c':=T(s,x)\cdot c\cdot
T(x,s)$ as follows. Let $x = u_0, u_1,\dotsc, u_k=x$ be the sequence
of vertices that belong to~$H$, while walking along~$c'$. The subpath
of~$c'$ between $u_{i-1}$~and~$u_{i}$ is homotopic to a subpath~$w_i$
of the facial walk~$f_H$ between occurrences of
$u_{i-1}$~and~$u_{i}$. Denote by~$w_i|_A$ the restriction of~$w_i$ to
the edges in~$A$. We have, in the above presentation
of~$\pi_1(\surf,s)$:
\[ [c'] = (w_1|_A)\cdot(w_2|_A)\dotsm (w_k|_A) 
\]
We call such a product a \define{term product} representation of~$[c']$ of
\define{height}~$k$. If we encode each term~$w_i|_A$ implicitly by two pointers,
in~$f_H|_A$, corresponding to its first and last edge, the above representation
can be stored as a list of $k$~pairs of pointers.
\begin{lemma}[\protect{see~\cite[Sec. 3.1]{dg-tcs-99}}]\label{lem:term-product}
  Let~$G$ be a graph of complexity~$n$ cellularly embedded on a
  surface~$\surf$.  We can preprocess $G$~and its embedding in
  $O(n)$~time such that the following holds. For any closed walk~$c$
  in~$G$ with $k$~edges, we can compute in $O(k)$~time a term product
  representation of height at most~$k$ of some closed walk freely
  homotopic to~$c$.
\end{lemma}
\begin{proof}
  From the preceding discussion, we just need to decompose the closed
  walk $c' = T(s,x)\cdot c\cdot T(x,s)$ in at most~$k$ terms, spending
  constant time per term. For this, we first compute a cut graph~$H$,
  a spanning tree~$T$ and the set~$A = E(H)\setminus E(T)$ as above.
  We also determine the facial walk $f_H = (e_1,e_2,\dotsc,e_m)$ of
  the embedding of~$H$. We consider $f_H$~as~a cyclic sequence of
  edges, meaning that $e_1$~is the successor of~$e_m$. To every
  (oriented) edge~$e_i$ along~$f_H$, we attach a pointer~$p(e_i)$
  toward the last previous edge in~$f_H$ that belongs to~$A$ and a
  pointer~$s(e_i)$ toward the first successive edge in~$A$. Thus,
  $p(e_i)$~and~$s(e_i)$ point to~$e_i$ if $e_i$~is in~$A$. For each
  $1\leq i\leq m$, let~$E_i$ be the set of oriented edges of~$G$
  sharing the head vertex of~$e_i$ and strictly between $e_i$~and~$e_{i+1}\inv$
  (taking indices modulo~$m$) in counterclockwise order
  around that head. To each oriented edge in~$E_i$ we attach a pointer
  toward~$p(e_i)$ with an \emph{exit} flag, and we attach to its
  opposite edge a pointer toward~$s(e_{i+1})$ with an \emph{enter}
  flag. An oriented edge with both endpoints in~$H$ will thus receive
  an exit and an enter flag. We finally attach the null pointer
  to the remaining oriented edges. This whole preprocessing step can
  clearly be performed in $O(n)$~time with the appropriate data
  structure for storing the embedding of~$G$. See
  Figure~\ref{fig:homotopy-encoding}.A for an illustration.
\begin{figure}
  \centering
  \includesvg[.85\linewidth]{homotopy-encoding}
  \caption{(A)~An oriented edge~$e_i$ of~$f_H$ with its two pointers
    $p(e_i)$~and~$s(e_i)$.\allowbreak\ (B)~The trace of a closed walk (dashed
    thick line) inside the disk bounded by~$f_H$. This walk is cut into two
    pieces respectively homotopic to $ab$~and~$b\inv aba\inv=1$.}
  \label{fig:homotopy-encoding}
\end{figure}

Given~$c$, we search along~$c$ for an edge~$e$ that is either an
entering edge or an edge in~$A$. If there is no such edge, $c$~is~either
contained in the interior of the facial disk bounded by~$f_H$
or contained in~$T$; in both cases, $c$~is contractible. Otherwise, we
define~$x$ as the tail vertex of~$e$. The terms of the decomposition
of~$c'$ can now be obtained by listing in order, starting from~$e$,
its edges in~$A$ (they define a term with a unique edge) and its
successive pairs of enter/exit edges along~$c$, using their attached
pointers, respectively $s(\cdot)$~and~$p(\cdot)$, for the encoding
(see Figure~\ref{fig:homotopy-encoding}.B). This clearly takes time
proportional to the length of~$c$ and produces at most $k$~terms.
\end{proof}

\section{The contractibility test}\label{sec:contractibility-test}
After reduction to a simplified framework, our test for the
contractibility of a closed curve~$p$ drawn on~$\surf$ relies on the
construction of a \emph{relevant region}~$\Pi_p$ in a specific tiling
of the universal cover~$\unicov$ of~$\surf$. This relevant region,
introduced by Colin~de~Verdi\`ere and Erickson~\cite{ce-tnspc-10}, contains
a lift of~$p$ and we shall compute that lift as we build~$\Pi_p$. We
can now decide whether~$p$ is contractible by just checking if its
lift is a closed curve in~$\unicov$. We detail our simplified
framework and tiling of~$\unicov$, as well as some of its
combinatorial properties, in
Section~\ref{subsec:simplified-framework}. We define the relevant
region and its construction in
Section~\ref{subsec:algo-contractibility}.

\subsection{A simplified framework}\label{subsec:simplified-framework}
\paragraph{Graphic interpretation of term products.}
Following Lemma~\ref{lem:term-product}, we can assume that $G$~is in
reduced form, with a single vertex~$s$ and a single face, and that the
input closed walks are given by their term product representations
in~$\langle E(G)\mathrel;f_G\rangle$. We denote by~$r$ the size of the
facial walk~$f_G$. Hence, $r=4g$~if $\surf$~is orientable with genus~$g$
and $r=2g$~if $\surf$~is a non-orientable surface of genus~$g$. We~can
view the cut surface~$\surf\cut G$ as an open regular $r$\dash/gon~$P$
whose boundary sides are labelled by the edges in~$f_G$. Obviously,
the boundary~$\partial P$ of~$P$ maps to~$G$ after gluing back the
sides of~$P$, and its vertices all map to~$s$. Each subpath of~$f_G$
labels one (or more) oriented subpath of~$\partial P$ which we can
associate with the chord between its endpoints. A term product
representation can thus be seen as a sequence of chords inside~$P$,
where each chord stands for a term. In order to represent these chords
as walks of constant complexity, we introduce an embedded graph~$H$
obtained as follows.  Consider the radial graph in~$P$ linking the
center of~$P$ to each vertex of~$\partial P$ along a straight segment;
this graph has $r+1$~vertices and $r$~edges.
\begin{figure}
  \centering
  \includesvg[.75\linewidth]{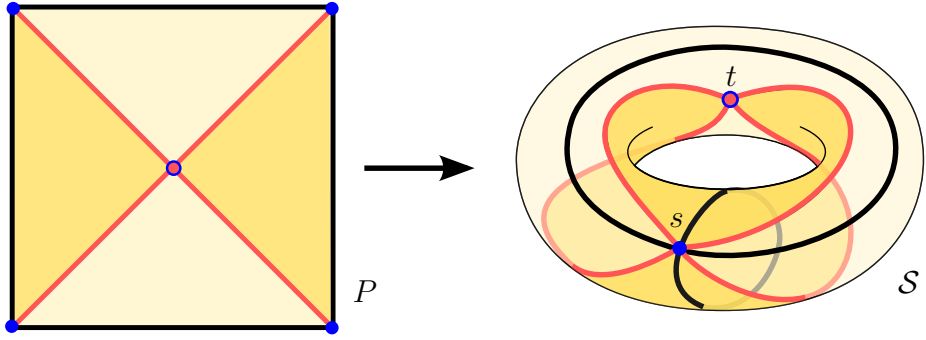}
  \caption{When $\surf$~is a torus, $P$~is a $4$\dash/gon and $r=4$. The
    black edges of~$\partial P$ projects to~$G$ in~$\surf$. Here~$G$
    is composed of two loops with basepoint~$s$. The four radial edges
    in~$P$ projects to the $s\mathord-t$~edges of the radial graph which
    possesses two faces.}
  \label{fig:radial-graph}
\end{figure}
After gluing back the boundary of~$P$ we obtain a bipartite graph~$H$
cellularly embedded on~$\surf$ with two vertices~$\left\{s,t\right\}$
and $r$~edges. Both vertices of~$H$ are $r$\dash/valent; each of the
$r/2$~faces of~$H$ is of length~$4$ and is cut in two ``triangles'' by
the unique edge of~$G$ it contains --- see
figure~\ref{fig:radial-graph}. We call $H$~the \define{radial graph}
of~$G$. A rotation system of~$H$ can be computed from that of~$G$ in
$\mcO(r)$~time.  Any chord in~$P$ is now homotopic to the $2$\dash/walk
with the same extremities and passing through~$t$. Consequently, if
$[c] = w_1 \dotsm w_k$ is a term product representation of height~$k$
stored as a list of pointers then in $\mcO(k)$~time we can obtain a
closed walk of length~$2k$ in~$H$, homotopic to~$c$.

\paragraph{Tiling of the universal cover.}
Let~$(\unicov,\pi)$ be the universal cover of~$\surf$. A loop of~$H$
is contractible if and only if its lift in~$\unicov$ is a loop and we
want to construct a finite part of~$\unicov$ large enough to contain
that lift. To this end we rely on a decomposition of~$\unicov$ similar
to the octagonal decomposition of~\cite{ce-tnspc-10}. Due to our
complexity constraints we cannot afford the cost of computing a tight
octagonal decomposition of~$\surf$. We
nevertheless present a tiling of~$\unicov$ with similar properties.

The lifted graph~$\uniH:=\pi\inv(H)$ of~$H$ is an infinite regular
bipartite graph with $r$\dash/valent vertices and $4$\dash/valent faces. A
vertex of $\uniH$ is said of \define{$s$\dash/type} or \define{$t$\dash/type} if
it respectively projects to $s$~or~$t$. Let~$H^*$ be the dual of~$H$
on~$\surf$ chosen so that its vertices lie in the middle of the edges
of~$G$. The lifted graph~$\uniH^*$ is the dual of~$\uniH$; it is an
infinite graph whose faces form a $\left\{ r,4 \right\}$\dash/tiling of the
universal cover, \ie/ four $r$\dash/gon faces meet at every vertex
of~$\uniH^*$. We say that two edges of~$H^*$ are \define{facing} each
other if they share an endpoint~$x$ and are not consecutive in the
circular order around~$x$. The four~edges meeting at any vertex
of~$H^*$ form two~pairs of facing edges. Facing edges are defined
similarly in~$\uniH^*$. The \define{line} induced by~$e^* \in
E(\uniH^*)$ is the smallest set~$\ell_{e^*} \subset E(\uniH^*)$
containing $e^*$~and the facing edge of any of its edges. Every line
is an infinite lift of some cycle of facing edges of~$H^*$, but
contrary to the tight cycles of the octagonal decomposition
in~\cite{ce-tnspc-10} this cycle may self-intersect. In fact, it can
happen that every line projects onto the \emph{same}
self-intersecting cycle of length~$r$. (As an example, one may consider
the case where the reduced graph~$G$ is embedded on the genus two
orientable surface with $f_G = aba\inv cb\inv d c\inv d\inv$.) To begin
with, we state a Dehn-like result about tilings. It will used in
Section~\ref{subsec:algo-contractibility} to bound
the size of the relevant region. 
\begin{lemma}
    \label{lem:dehn-argument}
    Let~$\Gamma$ be a regular $\{r,4\}$\dash/tiling of the plane (each face
    has $r$~edges and each vertex has $4$~neighbours) with~$r \geq
    4$. Every finite non-empty union~$R$ of faces of~$\Gamma$
    contains at least one face sharing $r-2$~consecutive
    edges with the boundary of~$R$.
\end{lemma}
\begin{proof}
  We follow the proof
  in~\cite[p. 188]{s-ctcgt-93} and gather the faces of~$\Gamma$ in
  concentric rings at increasing distance to some root face.  Let~$f$
  be a face of~$R$ in the outermost ring of~$R$. At least
  $r-3$~consecutive edges of~$f$ lie on the boundary of~$R$. Indeed,
  $f$~shares (at most) two edges with other faces of the same ring,
  and at most one edge with faces of the ring just beneath it. One
  more edge of~$f$ lies on the boundary of~$R$ if $f$~shares no edge
  with the inner ring, or if at least one of the neighbors of~$f$ in
  the outermost ring is not in~$R$. Else, we can replace~$f$ by any of
  its two ring neighbors, which do not share edges with the inner
  ring, to obtain a face with the required property. 
\end{proof}
Let again~$\Gamma$ be a regular $\{r,4\}$\dash/tiling of the plane with~$r\geq 4$. 
By the \name{Jordan}~curve theorem a closed simple path $c$
in the vertex-edge graph of $\Gamma$ bounds a finite
union $R$ of faces. A vertex $v$ of $c$ is called \define{convex},
\define{flat}, or
\define{reflex} if it is incident to respectively $0$, $1$, or $2$
edges interior to $R$. 
\begin{lemma}
  A closed simple path $c$ has at least $r$ convex vertices.
\end{lemma}
\begin{proof}
  We denote by $V_c, V_f$ and $V_r$ the respective number of convex,
  flat and reflex vertices of $c$. If $R$ is the region bounded by
  $c$, we also denote by $F$ its number of faces and by $V_i$ its
  number of interior vertices. We finally set $V:=V_i + V_c + V_f+
  V_r$ and call $E$ the total
  number of edges of the closed region $R$. By
  double-counting of the vertex-edge incidences, we get:
\[ 2E =  4V_i + 2V_c + 3V_f + 4V_r.
\]
By double-counting of the face-edge incidences, we get:
\[ rF = 2E - (V_c + V_f+ V_r).
\]
Euler's formula then implies $V-E+F=1$. Multiplying by $r$ we obtain:
\[r = rV - rE + rF = rV - (r-2)E - (V_c + V_f+ V_r), \]
and expending the values of $V$ and $E$:
\[ r = V_c - (r-4) V_i - \frac{r-4}{2}V_f - (r-3)V_r. \]
Since $r\geq 4$, we conclude that $V_c\geq r$.
\end{proof}
We can now state the main properties of lines in~$\uniH^*$.
\begin{proposition}
    \label{prop:lines:no-bigon-triangle}
    Suppose that $\surf$~is orientable with genus~$\geq 2$ or
    non-orientable with genus~$\geq 3$.  Then lines of~$\uniH^*$ do
    not self-intersect nor are cycles and two distinct lines intersect
    at most once. Moreover, the following properties hold.
    \begin{itemize}
        \item \textbf{(triangle-free Property)} Three pairwise distinct lines
            cannot pairwise intersect.
        \item \textbf{(quad-free Property)} Four pairwise distinct lines can
            neither form a quadrilateral, \ie/ include a possibly
            self-intersecting $4$\dash/gon.
    \end{itemize}
\end{proposition}
\begin{proof}
  With the above hypothesis $\uniH^*$~defines a regular
  $\{r,4\}$\dash/tiling of the plane with $r\geq 6$.  A contradiction
  of any of the above properties would imply the existence of a region
  bounded by at most four lines. Its boundary would contain a simple
  closed curve with at most four convex vertices, which is forbidden by
  the previous lemma. 
\end{proof}

\begin{lemma}
    \label{lem:lines:separating}
    Every line cuts~$\unicov$ in two infinite connected components.
\end{lemma}
\begin{proof}
  Let~$\ell$~be a line induced by an edge~$e^*$. We first show that
  $\unicov\setminus \ell$ has at most two connected
  components. Consider a small disk~$D$ centered at the midpoint of~$e^*$,
  so that $D\setminus \ell$~has two half-disk components $D_1$~and~$D_2$.
  Any point~$x$ of~$\unicov\setminus \ell$ is in the same
  component as one of $D_1$~or~$D_2$. Indeed, we can use an approach
  path from~$x$ to a point close to~$\ell$ and then follow a path
  along~$\ell$ on the same side as~$x$ until we meet~$D$, hence
  falling into $D_1$~or~$D_2$. It remains to prove that
  $\unicov\setminus \ell$~has at least two components. Consider the
  two endpoints $y$~and~$z$ of the edge~$e$ dual to~$e^*$. Let~$p$~be any
  path joining $y$~and~$z$ in $\unicov$. We just need to show that
  $p$~and~$\ell$ intersect. By pushing~$p$ along~$\uniH$, it is easily seen
  that $p$~is homotopic to a path~$q$ in~$\uniH$. Moreover, $p$~cuts~$\ell$
  if and only if $q$~does so. Since $e$~and~$q$ share the same
  endpoints, their projection $\pi(e)$~and~$\pi(q)$ are homotopic paths
  of~$H$ in~$\surf$. It is part of the folklore that homotopy in a
  cellularly embedded graph
  can be realized by combinatorial homotopies. Such homotopies are
  finite sequences of elementary moves obtained by replacing a subpath
  contained in a facial walk by the complementary subpath in that
  facial walk. Any combinatorial homotopy between $\pi(e)$~and~$\pi(q)$
  lifts to a combinatorial homotopy between $e$~and~$q$.
  Since the lift of an elementary move preserve the parity of the
  number of intersections with~$\ell$, we conclude that $q$~has an odd
  number of intersections with~$\ell$. It follows that $q$, whence~$p$,
  must cross~$\ell$.
\end{proof}

\subsection{Building the relevant region}\label{subsec:algo-contractibility}
\paragraph{The relevant region.}
Let~$p$~be a path in~$\uniH$. Following~\cite{ce-tnspc-10} we denote
by~$\region{p}$, and call the \define{relevant region} with respect
to~$p$, the union of closed faces of~$\uniH^*$ reachable from~$p(0)$
by crossing only lines crossed by~$p$, in any order. In other words,
if for any line~$\ell$ we denote by~$\ell^{+}$ the component
of~$\unicov \setminus \ell$ that contains~$p(0)$, then $\region{p}$~is
the ``convex polygon'' of~$\unicov$ formed by intersection of the
sets~$\ell \cup \ell^{+}$ for~all~$\ell$ \emph{not} crossed
by~$p$. See Figure~\ref{fig:poincare} for an illustration.
\begin{figure}
  \centering
  \includesvg[.4\linewidth]{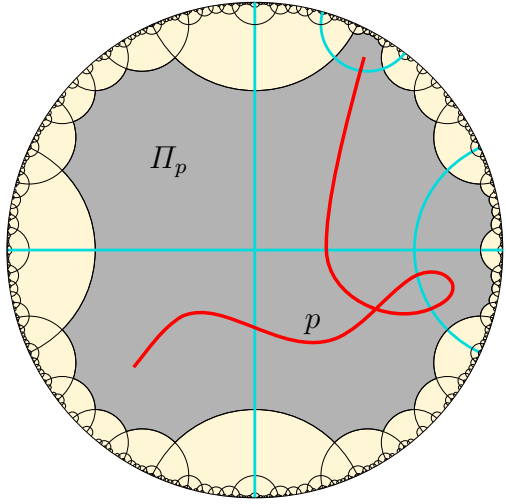}
  \caption{The Poincar\'e disk model of~$\unicov$ with lines
    represented as hyperbolic geodesics. Remark that the union of
    (light blue) lines crossed by~$p$ is not necessarily connected.}
  \label{fig:poincare}
\end{figure}
This region has interesting properties which make it easy and
efficient to build:
\begin{lemma}[\protect{\cite[lemma~4.1]{ce-tnspc-10}}]
    \label{lem:region:connected-inter-with-line}
    For any path~$p \subset \uniH$ and any line~$\ell$ the
    intersection~$\ell\cap\region{p}$ is either empty or a segment of~$\ell$
    whose relative interior is entirely included in either the interior
    or the boundary of~$\region{p}$.
\end{lemma}
\begin{proof}
  If $\ell\cap\region{p}$~is not empty, consider two
  points~$x,y\in\ell\cap\region{p}$. By definition of the relevant region,
  there is a path~$u$ joining $x$~to~$y$ (through~$p(0)$) whose
  relative interior is only crossed by lines also crossing~$p$. Any
  line crossing~$\ell$ between $x$~and~$y$ separates these two
  points, hence crosses~$u$, hence crosses~$p$. It easily follows that
  the open segment of~$\ell$ between $x$~and~$y$ is included in~$\region{p}$,
  either along its boundary or in its interior.
\end{proof}

\begin{lemma}[\protect{\cite[lemma~4.3]{ce-tnspc-10}}]
    \label{lem:region:small-number-of-faces}
    $\region{p}$~contains at most $\max(5 \len{p},1)$~faces of~$\uniH$.
\end{lemma}
\begin{proof}
  We recall that a vertex~$v$ of the boundary of~$\region{p}$ is
  \define{convex} if no line
  passing through~$v$ is crossed by~$p$. It is~\define{flat} if on
  the contrary one line through~$v$ is crossed by~$p$.
  By Lemma~\ref{lem:lines:separating}, a line that intersects the
  interior of~$\region{p}$ must cross~$p$. By
  Lemma~\ref{lem:region:connected-inter-with-line}, such a line
  intersect the boundary of~$\region{p}$ in at most two flat
  vertices. It follows that the number of flat vertices is at most~$2\len{p}$.
  If there is no flat vertex then $\region{p}$~is a single
  face, $\len{p} = 0$ and the lemma holds.  Else, note that between
  two consecutive flat vertices there are at most $r-2$~convex
  vertices, all on the boundary of a single face. Accordingly the
  boundary of~$\region{p}$ has at most $2(r-1)\len{p}$~edges.

  On the other hand, any union~$R$ of $k\geq 1$~faces has at least
  $(r-4)k$~edges on its boundary. Indeed, thanks to
  Lemma~\ref{lem:dehn-argument} we can recursively remove a face with
  at least $r-2$~boundary edges to decrease the perimeter by at
  least~$r-4$, until $R$~is empty. Applying this last result
  to~$\region{p}$ we get $(r-4)k \leq 2(r-1)\len{p}$, whence:
    \[ k \leq \frac{2(r-1)}{r-4} \len{p} \leq 5\len{p} \]
    the latter inequality stemming from the hypothesis~$r \geq 6$.
\end{proof}

The next result is crucial for building the relevant region:
\begin{lemma}[\protect{\cite[lemma~4.2]{ce-tnspc-10}}]
  \label{lem:region:construction}
  Let~$p$ be a path of~$\uniH$ and let $e$~be an edge with $p(1)$~as
  an endpoint.  Suppose $e$~crosses a line~$\ell$ not already crossed
  by~$p$. Then $\region{p} \cap \ell$~is a segment of the boundary
  of~$\region{p}$ along a connected set of
  faces~$V\subset\region{p}$. Moreover $\region{p\cdot e} = \region{p}
  \cup \Lambda$ where $\Lambda$~is the reflection of~$V$
  across~$\ell$.
\end{lemma}

\begin{proof}
  The intersection point of $e$~and~$\ell$ belong to the same (closed)
  face of~$\region{p}$ as~$p(1)$. The set $X := \region{p} \cap
  \ell$~is thus non-empty and, by
  lemma~\ref{lem:region:connected-inter-with-line},
  connected. Moreover, as
  $\ell$~does not cross~$p$, the line segment $X$~is part of
  the boundary of~$\region{p}$. In particular, each edge in~$X$ bounds
  some face of~$\region{p}$. These faces form a strip~$V$ where two
  consecutive faces have an interior edge of~$\region{p}$ in common.

  Each face in~$V$ has a bounding edge in~$X$ that also bounds a
  symmetric face outside~$\region{p}$. The union of these symmetric
  faces is the reflection~$\Lambda$ of~$V$. Furthermore
  $\region{p}\cup\Lambda \subset \region{p\cdot e}$ since one can
  reach any face of~$\Lambda$ from the symmetrical face of~$V$ just by
  crossing~$\ell$. To prove the reverse inclusion it is sufficient to
  show that none of the lines bounding~$\Lambda$ are crossed by~$p$. A
  single face~$f$ of~$\Lambda$ is bounded by~$\ell$, two \emph{inner}
  lines crossing~$\ell$ at vertices of~$f$, and $r-3$~other
  \emph{outer} lines. Since these two inner lines both cross~$\ell$,
  the triangle-free Property ensures that the two are
  disjoint and bound a \emph{band} of~$\unicov$ that contains~$f$. No
  line crosses~$\ell$ in this band, and a line crossing~$\ell$ cannot
  cross an inner line as the three would pairwise intersect ---
  which is again prohibited by
  the triangle-free Property. Since the
  outer lines bounding~$f$ have an edge in the band they cannot
  cross~$\ell$, hence cannot cross~$p$. Consider now an inner line~$\ell'$.
  It contributes to the boundary of~$\Lambda$ only when~$f$ is an
  extremal face of~$\Lambda$. In that case, the symmetrical face of
  $f$ in $V$ is also extremal and $\ell'$ is bounding
  $\region{p}$. So that $\ell'$~does not cross~$p$. We conclude the
  proof by noting that $\ell$~is the last line to consider on the boundary
  of~$\Lambda$.
\end{proof}

\paragraph{Computing the relevant region.}
Following lemma~\ref{lem:region:construction} we will build the
relevant region of the lift~$\unic$ of a loop~$c$ incrementally as we lift its
edges one at a time. If $p$~is the subpath of the lift~$\unic$ already
traversed and $\unie$~is the edge following~$p$ in~$\unic$ we
extend~$\region{p}$ to~$\region{p\cdot \unie}$ whenever $\unie$~crosses the
boundary of~$\region{p}$. From
Lemma~\ref{lem:region:small-number-of-faces}, we know that
$\region{p}$~contains $O(\len{p})$~faces. However, a naive
representation of~$\region{p}$ as a subgraph of~$\uniH^*$ with its
faces and edges would require $O(r\len{p})$~space, which we cannot
afford to obtain a linear time contractibility test. We rather store
the interior of~$\region{p}$ by its dual graph, \ie/ by the subgraph
of~$\uniH$ induced by the vertices dual to the faces of~$\region{p}$. More
precisely, to represent such a finite subgraph~$\Gamma$ of~$\uniH$ we
use an abstract data structure with the following operations:
\begin{itemize}
\item $\NEW(x)$, which returns and adds to~$V(\Gamma)$ a new
  vertex~$v$ with no neighbour, where $x \in
  \left\{s,t\right\}$~represents the projection~$\pi(v)$.
\item $\TYPE(v)$, which yields the projection~$\pi(v)$ for any
  vertex~$v\in V(\Gamma)$ --- that is $\TYPE(\NEW(x)) = x$.
\item $\JOIN(v,w,e)$, which adds a new edge~$\unie \in E(\Gamma)$ with
  endpoints $v$~and~$w$, where $v$~and~$w$ are existing vertices
  of~$\Gamma$ of different type and $e$~is an edge of~$H$
  representing~$\pi(\unie)$. Calling $\JOIN(v,w,e)$ and~$\JOIN(w,v,e)$
  is equivalent, and using any of those more than once for the same
  $v$~and~$e$ is unsupported.
\item $\NEXT(v,e)$, which finds the other endpoint of the lift of~$e$
  in~$\Gamma$, if it exists. More formally, it returns the unique
  vertex~$w\in\Gamma$ such that $\JOIN(v,w,e)$ has been called, if it
  exists, and a special value~$\NONE$ otherwise.
\end{itemize}
Given a loop~$c$ in~$H$, the following procedure then
constructs~$\region{\unic}$:
\begin{enumerate}
\item Create a variable~$v$ which will point to the \emph{current
    endpoint} of the partial lift of~$c$. With the operation $v
  \leftarrow \NEW(s)$, create a single vertex corresponding to the
  relevant region with respect to a null path.
\item \label{itm:computing-region:loop-step} Call $e$~the next edge to
  process in~$c$; if there is none left, exit.
\item \label{itm:computing-region:follow-step} If $\NEXT(v,e) \neq
  \NONE$, set $v \leftarrow \NEXT(v,e)$ and return to
  step~\ref{itm:computing-region:loop-step}.
\item \label{itm:computing-region:mirror-step} If $\NEXT(v,e) =
  \NONE$, the partial lift is exiting the current relevant region. In
  other words the lift~$\unie$~of~$e$ crosses a line~$\ell$ not
  crossed until now and we are in the situation of
  lemma~\ref{lem:region:construction}.  We need to enlarge~$\Gamma$ by
  performing a \emph{mirror operation}, creating and attaching
  vertices mirrored across~$\ell$ as suggested by the lemma. When
  done, $\NEXT(v,e) \neq \NONE$ and we can return to
  step~\ref{itm:computing-region:follow-step}.
\end{enumerate}
When we reach the end of~$c$ we have computed the relevant region of
its lift. Moreover, $c$~is contractible if and only if its lift is
closed, which reduces to the equality of the current endpoint~$v$ with
the first created endpoint. It remains to detail the mirror operation.

\paragraph{The mirror operation.}
Let~$v \in V(\Gamma)$ and~$e \in E(H)$ such that $\NEXT(v,e) =
\NONE$. Denote by~$\unie$ the lift of~$e$ from~$v$ in~$\uniH$ and by~$\ell$
the line crossed by~$\unie$. Let also $p$~be the partial lift
of~$c$ already processed; in particular $\Gamma$~is the subgraph of~$\uniH$
included in~$\region{p}$.
We begin the mirror sub-procedure by creating a vertex~$w
= \NEW(x)$, where $x \in \left\{ s,t \right\} \setminus \left\{
  \TYPE(v) \right\}$, and calling $\JOIN(v,w,e)$ so that $\unie$~is
now represented in~$\Gamma$.

We say that two edges of~$\uniH$ are \define{siblings} if their dual
edges are facing each other in~$\uniH^*$. Obviously two siblings cross
the same line and have no common endpoint. The four edges bounding any
face of~$\uniH$ form two pairs of siblings matching the two pairs of
facing edges around the corresponding vertex
of~$\uniH^*$. Let~$\lift{e_1}$ be one of the two siblings of
$\unie$. Denote by $\lift{e_0}$~and~$\lift{e_2}$ the other sibling
pair bounding the same face as $\unie$~and~$\lift{e_1}$, so that
$v$~is an endpoint of~$\lift{e_0}$ --- see
Figure~\ref{fig:mirror}.
\begin{figure}
  \centering
  \includesvg[.45\linewidth]{mirror}
  \caption{The path~$p$ is composed of two edges. The 
    graph~$\Gamma$ corresponding to its relevant region~$\region{p}$
    has four edges bounding a face of~$\uniH$.}
  \label{fig:mirror}
\end{figure}
If $\unicov[e_1]$~has an endpoint in~$\region{p}$, or equivalently if
$\unicov[e_0] \in E(\Gamma)$, then by
lemma~\ref{lem:region:construction} $\unie$,~$\unicov[e_0]$,
$\unicov[e_1]$ and~$\unicov[e_2]$ are all included in~$\region{p\cdot
  \unie}$. In that case $\unicov[e_1]$~and~$\unicov[e_2]$ need to be
added to~$\Gamma$. If on the contrary $\unicov[e_0] \notin E(\Gamma)$
then $v$~lies in a face of~$\uniH^*$ that is extremal in the chain~$V$
of lemma~\ref{lem:region:construction}, and the sibling~$\unicov[e_1]$
should not be added to~$\Gamma$. We proceed as follows: by walking
from~$\pi(v)$ around one face bounded by~$e$ in~$H$ we figure out the
projections $e_0$,~$e_1$ and~$e_2$ of $\unicov[e_0]$,~$\unicov[e_1]$
and~$\unicov[e_2]$ respectively. If $v_1 = \NEXT(v,e_0) \neq \NONE$
then we create $w_1 \leftarrow \NEW(\TYPE(v))$, and call
$\JOIN(v_1,w_1,e_1)$ and $\JOIN(v_1,w,e_2)$.  We handle similarly the
sibling of~$\unicov[e_1]$ which is not~$\unicov[e]$, walking around
the other face bounded by~$e_1$ and checking whether $v_1$~is extremal
in the chain~$V$, and so on until we reach the end of~$V$. There
remains to do the mirror in the other direction, starting from the
still unprocessed face of~$H$ bounded by~$e$.  Eventually, every
vertex dual to a face in the chain~$\Lambda$ of Lemma~\ref{lem:region:construction}
has been created, and that lemma ensures that we missed no
vertex. Since each time we add a vertex we also add all its edges
in~$\uniH$ linking to existing vertices of~$\Gamma$, we now have
that $\Gamma$ is the dual graph of $\region{p\cdot \unie}$.

\paragraph{Data structure.}
We present an implementation of the abstract graph structure that only needs
\emph{constant time} to perform any of its operations. We use a
technique inspired from~\cite[exercise~2.12~p.~71]{ahu-daca-74} taking
advantage of the RAM~model to allocate in $\mcO(1)$~time an
$r$\dash/sized segment of memory without initializing it. We begin by giving
integer indices between $1$~and~$r$ to edges of~$H$. The index of $e$ will be
denoted $\ID(e)$, and tables will be indexed from~$1$. Then:
\begin{itemize}
\item $\NEW(x)$ creates a vertex structure~$v$ with a field~\ttstyle{type}
  pointing to~$x$, an integer field~\ttstyle{count} with value~$0$,
  two uninitialized tables \ttstyle{index}~and~\ttstyle{rev} of
  $r$~integers, and an uninitialized table~\ttstyle{neighbor} of
  $r$~pointers;
\item $\TYPE(v)$ returns the value of~\ttstyle{type};
\item $\JOIN(v,w,e)$ increments~\ttstyle{count} in~$v$,
  points~$\ttstyle{neighbor}[\ttstyle{count}]$ towards~$w$, sets
  $\ttstyle{index}[\ttstyle{count}]\leftarrow\ID(e)$ and
  $\ttstyle{rev}[\ID(e)] \leftarrow \ttstyle{count}$, and affects~$w$
  similarly;
\item $\NEXT(v,e)$ returns
  $\ttstyle{neighbor}[\ttstyle{rev}[\ID(e)]]$ if we have $1 \leq
  \ttstyle{rev}[\ID(e)] \leq \ttstyle{count}$ together with
  $\ttstyle{index}[\ttstyle{rev}[\ID(e)]] = \ID(e)$, and returns
  $\NONE$ otherwise.
\end{itemize}
Of course if $\JOIN(v,w,e)$ has been called for some~$w$
then $\NEXT(v,e)$~will indeed return a pointer to~$w$. If not, then
$\ttstyle{rev}[\ID(e)]$ will still be uninitialized, and even if by
chance $1 \leq \ttstyle{rev}[\ID(e)] \leq \ttstyle{count}$ the
corresponding cell of \ttstyle{index} will have been filled by another
$\JOIN$~operation and the \emph{round-trip} check will fail.

\paragraph{Complexity.}
Checking and adding a sibling in the mirror subprocedure needs one
face traversal in~$H$, one call to~$\NEXT$ and if needed one call
to~$\NEW$ and two to~$\JOIN$.  The initialisation of the mirror
operation reduces to one call to~$\NEW$ and another one to~$\JOIN$. In
the end all mirrors will have used at most one~$\NEXT$, one~$\NEW$
and~two~$\JOIN$ operations for each vertex of $\Gamma$
--- that is each face of $\region{\lift{c}}$. Moreover,
step~\ref{itm:computing-region:follow-step} in the construction of
$\region{\unic}$ uses one~$\NEXT$ operation per edge of~$c$. Because
every operation takes $\mcO(1)$~time, and thanks to
lemma~\ref{lem:region:small-number-of-faces} our algorithm takes
$\mcO(\len{c})$~time. Taking into account the precomputation of
Lemma~\ref{lem:term-product} we have proved
Theorem~\ref{th:contractibility-test} when $r\geq 6$. This is the case when
$\surf$~is an orientable surface of genus at least~$2$ or when $\surf$~is
non-orientable with genus at least~$3$. In the remaining cases we
can expand the term product representations in $O(\len{c})$~time to
obtain a word in the computed presentation of~$\pi_1(\surf,s)$. The
word problem, that is testing if a word represents the unity in a
group presentation, is trivial in those cases. This is clear for the
torus or the projective plane since their fundamental group is
commutative. When $\surf$~is the Klein bottle, we can assume that the
computed presentation is $\langle a,b \mathrel; abab\inv\rangle$ by
applying an easy change of generators if necessary. The relator~$abab\inv$
allows us to commute $a$~and~$b$ up to an inverse and to get
a canonical form $a^ub^v$ in $O(\len{c})$~time. This in turn solves the
word problem in linear time.

\section{The free homotopy test}\label{sec:free-homotopy-test}
We now tackle the free homotopy test. We restrict to the case where
$\surf$~is an orientable surface of genus at least two.  We can thus
orient the cellular embedding. This amounts to give a preferred
traversal direction to each facial walk. In particular, every oriented
edge~$e$ belongs to exactly one such facial walk, which we designate
as the left face of~$e$. This allows us in turn to associate with~$e$ a
dual edge oriented from the left face to the right face of~$e$. This
correspondence between the oriented edges of the embedded graph and
its dual will be implicit in the sequel.

We want to decide if two cycles $c$~and~$d$ on~$\surf$ are freely
homotopic. After running our contractibility test on $c$~and~$d$, we
can assume that none of these two cycles is contractible.  From
Lemma~\ref{lem:term-product} and the discussion of the simplified
framework in Section~\ref{subsec:simplified-framework}, we can also
assume that $c$~and~$d$ are given as closed walks in the radial graph~$H$.
Let $(\cycov{c},\pi_c)$~be the $c$\dash/cyclic cover of~$\surf$.
Recall that $\cycov{c}$~can be viewed as the orbit space of
the action of~$\left<\trans\right>$ where $\trans$~is the unique
automorphism of~$(\unicov,\pi)$ sending $\unic(0)$~on~$\unic(1)$ for a
given lift $\unic$ of~$c$. We will refer to~$\trans$ as a
\emph{translation} of~$\unicov$, as it can indeed be realized as a
translation of the hyperbolic plane. Notice that $\surf$~being
orientable, $\trans$~is orientation preserving. The
projection~$\cyproj$ sending a point of~$\unicov$ to its orbit
makes~$(\unicov,\cyproj)$ a covering space of~$\cycov{c}$ with $\pi =
\pi_c\circ\cyproj$. We denote by $\cyH$~and~$\cyH^*$ the respective
lifts of $H$~and~$H^*$ in~$(\cycov{c},\pi_c)$. In the sequel,
regularity of paths in $\surf$,~$\unicov$ or~$\cycov{c}$ is considered
with respect to $H^*$,~$\uniH^*$ and~$\cyH^*$ respectively, and so are
the \CW/s of regular paths.
\subsection{Structure of the cyclic cover}
Recall that a line in~$\unicov$ is an infinite sequence of facing
edges in~$\uniH^*$. We start by stating some structural properties of
lines. Most of these properties appear in~\cite{ce-tnspc-10} in one
form or another. However, due to our different notion of lines, we
cannot rely on the proofs in there. For instance, a line as
in~\cite{ce-tnspc-10} projects to a simple curve in the $c$\dash/cyclic
cover~$\cycov{c}$. In our case, though, a line may project to a
self-intersecting curve in $\cycov{c}$. In fact, as previously noted,
it may be the case that all of our lines have the same projection on~$\surf$.

\paragraph{$\trans$\dash/transversal lines.}
Let $\ell$~be a line such that $\ell \cap \tau(\ell) = \emptyset$ and
denote by~$\mathring B_{\ell}$ the open band of~$\unicov$ bounded by
$\ell$~and~$\tau(\ell)$.  The line~$\ell$ is said
\define{$\tau$\dash/transversal} if $B_{\ell}:=\mathring
B_{\ell}\cup\ell$ is a fundamental domain\footnote{\ie/ every orbit
  of $\left<\tau\right>$ has a unique representative in $B_{\ell}$.}
for the action of~$\left<\tau\right>$
over~$\unicov$. Equivalently, $\ell$~is $\tau$\dash/transversal if
$\unicov$~is the disjoint union of all the translates of~$B_{\ell}$.
In such a case we can obtain~$\cycov{c}$ by point-wise
identification of the boundaries $\ell$~and~$\trans(\ell)$
of~$B_\ell$. The following proposition gives a characterisation
$\trans$\dash/transversal lines whose existence are stated in
Proposition~\ref{prop:transversal-existence}. 
\begin{proposition}\label{prop:lines:transversal-condition}
  Let $\ell$~be a line such that $\ell \cap \tau(\ell) =
  \emptyset$. Then, $\ell$~is $\trans$\dash/transversal if and only if
  there exists~$x \in \unicov$ such that $\ell$~separates
  $x$~from~$\tau(x)$ but $\tau^{-1}(\ell)$~and~$\tau(\ell)$ do not.
\end{proposition}
\begin{proof}
  If $B_{\ell}$~is a fundamental domain, then
  $\mathring B_{\ell}$~and~$\tau\inv(\mathring B_{\ell})$
  are disjoint and separated by~$\ell$.
  The direct implication in the equivalence of the lemma is
  thus trivial. For the reverse implication consider an~$x$ as in the
  proposition. Put $x_i := \tau^i(x)$~and $\ell_i := \tau^i(\ell)$. By
  assumption, for~all~$i$, $\ell_i$~separates $x_i$~from~$x_{i+1}$
  while $\ell_{i-1}$~and~$\ell_{i+1}$ do not. Let $U_i$~and~$L_i$ be
  the two
  components of~$\unicov\setminus\ell_i$ with $x_{i+1}\in U_i$~and
  $x_{i}\in L_i$. From the assumption, we also have for~all~$i$ that
  $x_{i}\in L_{i+1}$ and $x_{i+1}\in U_{i-1}$. In particular,
  $L_i$~contains both $x_i$~and~$x_{i-1}$. Since $\ell_{i-1}$ separates
  these two points it must be that $\ell_{i-1}\subset L_i$. Likewise,
  we have $\ell_{i+1}\subset U_i$. In other words, we have $L_{i-1} =
  \tau\inv(L_{i})\subset L_i$ and $U_{i+1} = \tau(U_i)\subset U_i$.
  We~finally set $B_i:= (L_{i}\cap U_{i-1}) \cup \ell_{i-1}$. See
  Figure~\ref{fig:transversal-character}.
  \begin{figure}
    \centering
    \includesvg[.4\linewidth]{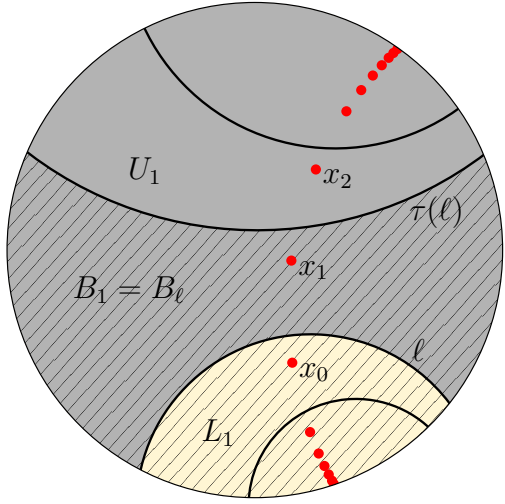}
    \caption{The fundamental domain~$B_1$ is the intersection of the
      grey upper region~$U_1$ with the hatched lower region~$L_1$.}
    \label{fig:transversal-character}
  \end{figure}

  We claim that the~$B_i$'s are pairwise disjoint. 
  Indeed, for~$i<j$, the inclusions $\ell_{i-1}\subset L_i\subset L_{j-1}$ imply
  $\ell_{i-1}\cap B_j=\emptyset$ since $L_{j-1}$~and~$U_{j-1}$ are
  disjoint. Likewise, the inclusions $\ell_{j-1}\subset U_{j-2}\subset
  U_{i-1}$ imply $\ell_{j-1}\cap B_i=\emptyset$. We thus obtain
  \[ B_i\cap B_j = (L_{i}\cap U_{i-1})\cap (L_{j}\cap U_{j-1}) =
  L_i\cap U_{j-1} \subset L_{i}\cap U_{i} =\emptyset.
  \]
  Moreover, the union of the~$B_i$'s cover~$\unicov$. To see this,
  consider any point~$y$ in~$\unicov$ contained in the closure of some
  face~$f_y$ of~$\lift{H}^*$. Let~$p$~be a path in~$\lift{H}$ between
  the vertex dual to~$f_y$ and the vertex dual to a face that contains~$x_1$.
  If~$p$~does not cross any~$\ell_i$, then $p$, hence~$y$, must
  be contained in~$B_1$. Otherwise, let $z$~be the first intersection
  point along~$p$ with~$\cup_i\ell_i$. By considering the appropriate
  translate of the subpath of~$p$ between $y$~and~$z$, we see that
  some translate of~$y$ is contained in~$B_1$. We conclude the lemma
  by noting that $B_{i+1} = \tau(B_i)$, whence $\unicov$~is the
  disjoint union of the translates of~$B_1 = B_{\ell}$.
 \end{proof}

\paragraph{Generators of the cyclic cover.}
A loop of~$\cycov{c}$ that is regular (for~$\cyH^*$) and freely
homotopic to~$\cyproj(\unic)$ is called a \define{generator} of~$\cycov{c}$.
Since $\pi_c(\cyproj(\unic)) = \pi(\unic) = c$, every
generator projects on~$\surf$ to a loop that is freely homotopic
to~$c$. Conversely, every regular loop freely homotopic to~$c$ has a
lift in~$\cycov{c}$ which is a generator. A~\define{minimal generator}
is a generator whose \CW/ (with respect to~$\cyH^*$) is minimal among
generators; it projects to a regular loop of \MCW/ in the free
homotopy class of~$c$.
\begin{lemma}
    \label{lem:generator:lift-is-tau-compatible}
    Let~$\lift{\gamma}\subset\unicov$~be any lift of a generator~$\gamma$
    of~$\cycov{c}$. Then, $\trans$~sends
    $\lift{\gamma}(0)$~on~$\lift{\gamma}(1)$.
\end{lemma}
\begin{proof}
  By definition of a generator there exists a homotopy $h: [0,1]^2
  \rightarrow \cycov{c}$ from~$\cyproj(\unic)$ to~$\gamma$. Consider
  in~$\unicov$ the unique lift~$\unih$ of~$h$
  with~$\unih(1,0)=\lift{\gamma}(0)$ (see
  Figure~\ref{fig:lift-translation}).
  \begin{figure}
    \centering
    \includesvg[.8\linewidth]{lift-translation}
    \caption{(Left)~$\cyproj(\unic)$~and~$\gamma$ are homotopic as
      generators of~$\cycov{c}$.\allowbreak\ (Right)~Any homotopy~$h$ lifts
      to a homotopy in~$\unicov$.}
    \label{fig:lift-translation}
  \end{figure}
Then $\unih(1,\cdot) =
  \lift{\gamma}$ and $\unih(0,\cdot)$~is some lift
  of~$\cyproj(\unic)$. Since $\cycov{c}$~is a cylinder, the
  automorphism group of~$(\unicov,\phi_c)$ is the cyclic group
  generated by~$\trans$. It follows that $\unih(0,\cdot) =
  \trans^i(\cyproj(\unic))$ for some integer~$i$. In particular,
  $\unih(0,1) =
  \trans\mathord{\left(\smash{\unih}(0,0)\right)}$. Since
  $h(\cdot,0)=h(\cdot,1)$, it ensues by uniqueness of lifts that
  $\unih(.,1) =
  \trans\mathord{\left(\smash{\unih}(\cdot,0)\right)}$. Whence
  $\unih(1,1) = \trans\mathord{\left(\smash{\unih}(1,0)\right)}$ which
  is exactly $\lift{\gamma}(1) =
  \trans\mathord{\left(\smash{\lift{\gamma}}(0)\right)}$.
\end{proof}
It follows that a regular path joining two points $x$~and~$y$ in~$\unicov$
is a lift of a generator if and only~if $y =\trans(x)$.
Consider a minimal generator~$\gamma$ and one of its lifts~$\lift{\gamma}
\subset\unicov$. Its reciprocal image~$\ell_{\gamma}:=
\cyproj\inv(\gamma)$ is the curve obtained by concatenation of all the
translates~$\trans^i(\lift{\gamma})$, $i\in\Z$, of~$\lift{\gamma}$. Remark that,
as far as intersections of~$\lift{\gamma}$ with lines is concerned, we
can assume that $\gamma$, hence $\ell_{\gamma}$, is simple. Indeed, if
$\gamma$~self-intersects it must contain by Lemma~\ref{lem:cylinder-2} a
subpath forming a contractible loop; by minimality of~$\gamma$ this
subpath does not intersect any line and we can cut it off without
changing the intersections of~$\gamma$ with any line. By induction on
the number of self-intersections, we can thus make $\gamma$ simple.
The simple curve~$\ell_{\gamma}$ actually behaves like a line in~$\unicov$:
\begin{lemma}\label{lem:l-gamma-separating}
  $\ell_{\gamma}$~is separating in~$\unicov$.
\end{lemma}
\begin{proof}
  Since $\gamma$~is already separating on~$\cycov{c}$, this is
  certainly the case for its reciprocal image~$\ell_{\gamma}$.
\end{proof}
\begin{lemma}\label{lem:l_gamma-line-crossing}
  A line can intersect~$\ell_{\gamma}$ at most once.
\end{lemma}
\begin{proof}
  Suppose that a line~$\ell$ crosses~$\ell_{\gamma}$ at least
  twice. Then, there exists a disk bounded by a subpath of~$\ell$ and
  a subpath of~$\ell_{\gamma}$ that only intersect at their common
  extremities. Let~$D$~be an inclusion-wise minimal such disk. Call~$p$~the
  subpath of~$\ell$ bounding~$D$ and call $q$~the
  subpath of~$\ell_{\gamma}$ bounding~$D$. Let $x$~and~$y$ be the common
  extremities of $p$~and~$q$ so that $x$~occurs before~$y$
  along~$\ell_{\gamma}$ (oriented from~$x$ to~$\trans(x)$). Note that
  $\ell$~cannot cross~$q$ as this would otherwise contradict the minimality
  of~$D$. Since $\trans$~is orientation preserving, $p$~and its
  translates~$\trans^i(p)$ occur on the same side of~$\ell_{\gamma}$
  (see previous lemma). We shall consider all the relative positions
  of~$y$ with respect to $x$,~$\trans(x)$ and~$\trans^2(x)$
  along~$\ell_{\gamma}$ and reach a contradiction in each
  case. Figure~\ref{fig:line-l-gamma} illustrates each of the cases.
  \begin{figure}
    \centering
    \includesvg[.85\linewidth]{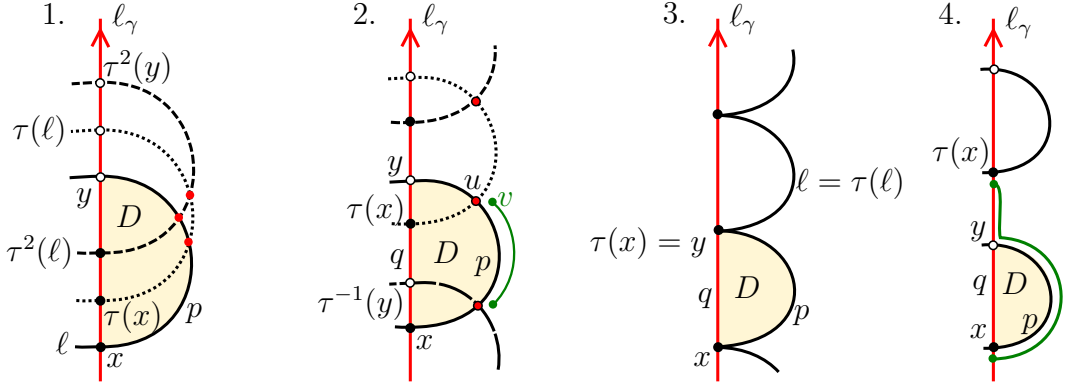}
    \caption{Along $\ell_{\gamma}$, the point~$y$ lies
        (1.)~beyond $\trans^2(x)$,\allowbreak\ (2.)~strictly between
        $\trans(x)$~and~$\trans^2(x)$,\allowbreak\ (3.)~on $\trans(x)$
        and~(4.)~strictly between $x$~and~$\trans(x)$.}
    \label{fig:line-l-gamma}
  \end{figure}
  \begin{enumerate}
  \item If $y=\trans^2(x)$, or if $\trans^2(x)$~separates $x$~from~$y$
    on~$\ell_{\gamma}$, then $\trans(\ell)$~and~$\trans^2(\ell)$
    cross~$q$ and also cross~$p$ by minimality of~$D$. It follows that
    $\ell$,~$\trans(\ell)$ and~$\trans^2(\ell)$ would pairwise intersect,
    in contradiction with the triangle-free Property.
  \item If $y$~strictly lies between $\trans(x)$~and~$\trans^2(x)$, then
    $p$~and~$\trans(p)$ must intersect. Let~$u$~be their intersection point
    (there is a unique one since we cannot have
    $\ell=\trans(\ell)$). We introduce the notation~$\alpha[u,v]$ for
    the subpath of a simple path~$\alpha$ between two of its points
    $u$~and~$v$. The concatenation
    \[
    p[\trans\inv(u),u]\cdot\trans(p)[u,\trans(x)]\cdot
    q[\trans(x),\trans\inv(y)]\cdot
    \trans\inv(p)[\trans\inv(y),\trans\inv(u)]
    \]
    is a simple closed path. It follows from
    the triangle-free Property that any line
    intersecting~$p[\trans\inv(u),u]$ must also
    intersect~$q[\trans(x),\trans\inv(y)]$. Consider a point~$v$ near~$u$ and
    outside~$D\cup\trans(D)$ and consider a regular path~$p'$
    between $\trans\inv(v)$~and~$v$, close to~$p[\trans\inv(u),u]$ but
    disjoint from~$p$; we may choose~$p'$ in such a way that it
    crosses the same line as the relative interior of~$p[\trans\inv(u),u]$.
    It follows that the \CW/ of~$p'$ is strictly
    less than the \CW/ of the regular path~$q[x',\trans(x')]$, where
    $x'\in q$~is a point close to~$x$. This contradicts the fact that
    $\gamma$~is a minimal generator since $q[x',\trans(x')]$~projects
    to~$\gamma$ and $p'$~projects to a shorter generator of~$\cycov{c}$.
  \item If $y=\trans(x)$ then $\gamma$~and~$\pi_c(\ell) = \pi_c(p)$
    are homotopic and intersect once. Since $\cycov{c}$~is a cylinder,
    this implies that $\gamma$~and~$\pi_c(\ell)$ touch but do not
    cross transversely. This in turn contradicts the regularity of~$\gamma$.
  \item Finally, if $y$~strictly separates $x$~from~$\trans(x)$ then
    the path obtained by following~$\ell$ between $x$~and~$y$ and then
    $\ell_{\gamma}$ between $y$~and~$\trans(x)$ provides, after 
    infinitesimal perturbation, a lift of a generator with \CW/
    strictly less than~$\gamma$. This contradicts the minimality of~$\gamma$.
  \end{enumerate}
\end{proof}
\begin{lemma}\label{lem:l_gamma-no-triangle}
  If a line~$\ell$ intersects~$\ell_{\gamma}$, then $\ell\cap
  \trans(\ell) = \emptyset$.
\end{lemma}
\begin{proof}
  Suppose that $\ell$~intersects both $\ell_{\gamma}$~and~$\trans(\ell)$.
  From Lemma~\ref{lem:l_gamma-line-crossing}, $\ell$~and~$\ell_{\gamma}$
  cross exactly once. In particular, $\ell$~and~$\trans(\ell)$
  cannot coincide, hence must cross exactly once. Let~$x:=\ell \cap
  \ell_{\gamma}$ and let~$y := \ell\cap\trans(\ell)$. Since
  $\trans$~preserve the orientation, $y$~and~$\trans(y)$ must be
  on the same side of~$\ell_{\gamma}$. The
  point~$\trans(y)$ cannot lie on the same side of~$\ell$ as~$\trans(x)$.
  Indeed, $\trans^2(\ell)$~would then enter the triangle~$(x,\trans(x),y)$
  formed by $\ell_{\gamma}$,~$\trans(\ell)$ and~$\ell$. From the preceding
  Lemma~\ref{lem:l_gamma-line-crossing}
  and from the triangle-free Property,
  $\trans^2(\ell)$~could not exit that triangle, leading to a
  contradiction. But $\trans(y)$~can neither lie on the other side of~$\ell$
  as $\trans^2(\ell)$~would have to cross both $\ell$~and~$\trans(\ell)$,
  contradicting the triangle-free Property.
\end{proof}
\begin{proposition}
   \label{prop:transversal-existence}
   Let~$\gamma$~be a minimal generator whose basepoint is not on any
   line. Let $\lift{\gamma}$~be a lift of~$\gamma$ in~$\unicov$.  Any
   line~$\ell$ crossed by~$\lift{\gamma}$ is
   $\trans$\dash/transversal. In particular, there exists a
   $\trans$-transversal line.
\end{proposition}
\begin{proof}
  Put~$x := \lift{\gamma}(0)$. By
  Lemma~\ref{lem:generator:lift-is-tau-compatible}, $\lift{\gamma}(1)
  = \trans(x)$. By Lemma~\ref{lem:l_gamma-line-crossing}, any line~$\ell$
  crossing~$\lift{\gamma}$ separates $x$~from~$\trans(x)$ while
  $\trans\inv(\ell)$~and~$\trans(\ell)$ do not. Together with
  Lemma~\ref{lem:l_gamma-no-triangle}, this allows us to apply
  Proposition~\ref{prop:lines:transversal-condition} and to conclude that
  $\ell$~is $\trans$\dash/transversal. Since $\lift{\gamma}$~must be
  crossed by at least one line, this finally prove the existence of a
  transversal line. Indeed, if $\lift{\gamma}$~was not crossed by any
  line it would lie inside a single face of~$\uniH^*$: its
  projection~$\pi(\lift{\gamma}) \subset \surf$ would also stay in a
  single face of~$H^*$ and be trivially contractible, contradicting
  the fact that $\pi(\lift{\gamma})$~is freely homotopic to~$c$.
\end{proof}

\begin{cor}
    \label{cor:transversal-unic}
    There exists a $\trans$\dash/transversal
    line that separates the endpoints of~$\unic$ from each other.
\end{cor}
\begin{proof}
  By Proposition~\ref{prop:transversal-existence}, there exists a
  $\trans$\dash/transversal line~$\ell$. Since $\unic(1)$~is a
  translate of~$\unic(0)$, these two endpoints are
  included in two successive fundamental domains determined by~$\ell$
  and its translates. Since the basepoint~$s$ of~$c$ is not on the
  projection of any line, the translate of~$\ell$ that separates the
  interior of these
  two fundamental domains also separates $\unic(1)$~from~$\unic(0)$.
\end{proof}

\paragraph{$\trans$\dash/invariant line.} A line~$\ell$ such that
$\trans(\ell)=\ell$ is said \define{$\trans$\dash/invariant}. Note that
this equality does not hold pointwise, but globally.
\begin{lemma}\label{lem:1-invariant-line}
  There is at most one $\trans$\dash/invariant line.
\end{lemma}
\begin{proof}
  Let $\ell$~and~$\ell'$ be two $\trans$\dash/invariant lines. Since a
  $\trans$\dash/invariant line cannot lie in a single fundamental domain,
  any $\trans$\dash/transversal line~$\lambda$ must intersect
  $\ell$~and~$\ell'$. If $\ell$~and~$\ell'$ where distinct then they
  would form a quadrilateral with $\lambda$~and~$\trans(\lambda)$,
  in contradiction with the quad-free Property.
\end{proof}
\begin{lemma}\label{lem:3-translates-invariant}
  A line that intersects three consecutive translates of a
  $\trans$\dash/transversal is $\trans$\dash/invariant.
\end{lemma}
\begin{proof}
  Suppose that a line~$\ell$ intersects a $\trans$\dash/transversal
  line~$\lambda$ as well as $\trans(\lambda)$~and~$\trans^2(\lambda)$. If
  $\ell\neq\trans(\ell)$ then $\ell$,~$\trans(\ell)$,
  $\trans(\lambda)$ and~$\trans^2(\lambda)$ form a quadrilateral, in
  contradiction with the quad-free Property.
\end{proof}
\subsection{The canonical generator}
Since~$\surf$ is oriented we can speak of the left or
right side of a minimal generator. Our aim is to prove that the set of
minimal generators of~$\cycov{c}$ covers a bounded cylinder allowing
us to define its right boundary as a canonical representative of the
free homotopy class of~$c$. By definition, a $\trans$\dash/transversal line
projects in~$\cycov{c}$ to a simple curve. We call this projection a
\define{$c$\dash/transversal}. A $c$\dash/transversal~$\ell$ crosses exactly
once every minimal generator~$\gamma$. Indeed, if $\ell$~and~$\gamma$
had two intersections $x$~and~$y$, the subpath of~$\ell$ between $x$~and~$y$
would be homotopic to one of the two paths cut by $x$~and~$y$
along~$\gamma$ (see Lemma~\ref{lem:cylinder-1}). These two homotopic
subpaths would lift in~$\unicov$ to a closed path, implying that a
lift of~$\ell$ cuts $\ell_{\gamma}$~twice; a contradiction with
Lemma~\ref{lem:l_gamma-line-crossing}. Moreover,
Proposition~\ref{prop:transversal-existence} implies that any minimal
generator is crossed by $c$\dash/transversals only. The number of
$c$\dash/transversals in~$\cycov{c}$ is thus equal to the length of the
minimal generators which is in turn no larger than~$\len c$. Notice
that the orientation of~$\surf$ and of the minimal generators induce a
left-to-right orientation of the $c$\dash/transversals.

\begin{lemma}\label{lem:uncross-generators}
  Let $\mu$~and~$\nu$ be two minimal generators. There exist two
  disjoint and simple minimal generators $\gamma$~and~$\sigma$ such
  that the set of edges of~$\cyH^*$ crossed by~$\gamma\cup \sigma$ is
  the same as the set of edges crossed by~$\mu\cup \nu$.
\end{lemma}
\begin{proof}
  As remarked above Lemma~\ref{lem:l-gamma-separating}, we can assume
  that both $\mu$~and~$\nu$ are simple. Suppose that they
  intersect. If necessary, we can modify $\mu$~and~$\nu$ inside the
  interior of each face and edge so that they cross only inside faces
  and only transversely. Since $\mu$~is separating, it must intersect~$\nu$
  in at least two points. These two points, say $x$~and~$y$, cut
  each generator in two pieces, say $\mu = a\cdot b$~and $\nu = c\cdot
  d$. By Lemma~\ref{lem:cylinder-2} we may assume without loss of
  generality, that $a$~is homotopic to~$c$ and that $a\cdot d$~is a
  generator (see Figure~\ref{fig:separates-generators}). 
  \begin{figure}
    \centering
    \includesvg[.3\linewidth]{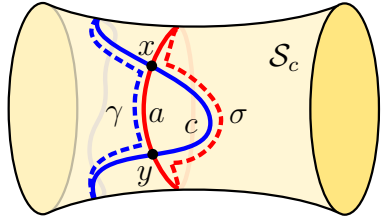}
    \caption{The intersecting generators $\mu$~and~$\nu$ (thick plain
      line) may be replaced by the disjoint generators $\gamma$~and~$\sigma$
      (thick dashed line).}
    \label{fig:separates-generators}
  \end{figure}
  By minimality
  of $\mu$~and~$\nu$, the paths $a$~and~$c$ must have the same \CW/
  and similarly for $b$~and~$d$. We can thus replace $\mu$~and~$\nu$
  by the minimal generators obtained by slightly perturbing the
  concatenations $a\cdot d$~and $c\cdot b$ so as to remove the two
  intersections at $x$~and~$y$. We claim that $a\cdot d$~and $c\cdot b$
  have fewer intersections than $\mu$~and~$\nu$. Indeed, denoting
  by~$\mathring e$ the relative interior of a path~$e$, we have
  \[
  |\mu\cap\nu| = |\mathring a\cap \mathring c| + |\mathring a\cap
  \mathring d| + |\mathring b\cap \mathring c| + |\mathring b\cap
  \mathring d| + 2
  \]
  while
  \[ 
  |a\cdot d \cap c\cdot b| = |\mathring a\cap \mathring c| +
  |\mathring a \cap \mathring b| + |\mathring d\cap \mathring c| +
  |\mathring d\cap \mathring b| = |\mathring a\cap \mathring c| +
  |\mathring d\cap \mathring b|,
  \]
  since, by simplicity of $\mu$~and~$\nu$: $|\mathring a \cap
  \mathring b| = |\mathring d\cap \mathring c| = 0$. The new minimal
  generators $a\cdot d$~and $c\cdot b$ obviously cross the same edges
  as $\mu$~and~$\nu$. They may now self-intersect, but as initially
  noted, we can remove contractible loops until they become simple. We
  conclude the proof with a
  simple recursion on $|\mu\cap\nu|$.
\end{proof}
We now consider two disjoint and simple minimal generators
$\gamma$~and~$\sigma$. They bound an annulus~${\mathcal A}$ in~$\cycov{c}$.
Since $\gamma$~and~$\sigma$ are crossed by $c$\dash/transversal curves only, a
line~$\ell$ of~$\unicov$ whose projection~$\cyproj(\ell)$
intersects~$\mathcal A$ is either $\trans$\dash/transversal or
$\trans$\dash/invariant.
Indeed, if $\ell$~is not $\trans$\dash/transversal, $\cyproj(\ell)$~must
stay in the finite subgraph of~$\cyH^*$ interior to~$\mathcal A$; it
follows that $\cyproj(\ell)$~uses some edge twice, which can only happen
if $\ell$~is $\trans$\dash/invariant by
Proposition~\ref{prop:lines:no-bigon-triangle}. In this latter case,
$\cyproj(\ell)$~is a simple generator and by
Lemma~\ref{lem:1-invariant-line}, there is only one such curve in
$\cycov{c}$. Note that this curve is crossed once by every
$c$\dash/transversal, hence is composed of $|\gamma|$~edges. We first
bound the complexity of~${\mathcal A}$.
\begin{lemma}\label{lem:short-bound}
  Let $V_I$,~$E_I$ and~$F$ be the respective numbers of vertices,
  edges and faces of~$\cyH^*$ intersected by~$\mathcal A$.
  Then $V_I \leq |\gamma|$, $E_I \leq 3|\gamma|$~and $F \leq 2|\gamma|$.
\end{lemma}
\begin{proof}
  If ${\mathcal A}$~contains the projection~$\cyproj(\ell)$ of a
  $\trans$\dash/invariant line~$\ell$, we consider two minimal generators
  $\lambda$~and~$\rho$ running parallel to~$\cyproj(\ell)$ and
  respectively to the left and to the right of~$\cyproj(\ell)$.
  Assuming that $\gamma$~is to the left of~$\cyproj(\ell)$,
  we can bound the complexity of~$\cyH^*$ in~$\mathcal A$
  by the complexity of~$\cyproj(\ell)$ plus the complexity of~$\cyH^*$
  inside the two annuli bounded by $\gamma$~and~$\lambda$, and by
  $\rho$~and~$\sigma$ respectively.

  We can now assume that $\mathcal A$~is crossed by $c$\dash/transversals
  only. These $c$\dash/transversals form an arrangement of~$|\gamma|$
  curves in~$\mathcal A$ where two $c$\dash/transversals can cross at most
  twice by Lemma~\ref{lem:3-translates-invariant}. Together with
  $\gamma$~and~$\sigma$, this arrangement defines a subdivision
  of~$\mathcal A$ whose number of boundary vertices is~$2|\gamma|$. We
  distinguish two cases according to whether $c$\dash/transversals pairwise
  intersect at most once or twice.
  \paragraph{Case where each pair of $c$\dash/transversals intersects at most
    once.}
  If the $c$\dash/transversals are pairwise disjoint in~$\mathcal A$, the lemma
  is trivial: $\mathcal A$~contains no vertex and intersects $|\gamma|$~edges
  and as many faces. Otherwise, we first establish a
  correspondence between faces and interior vertices.

  Consider an interior vertex~$x$ of the subdivision; it is the
  intersection of two $c$\dash/transversals $u$~and~$v$ (refer to
  Figure~\ref{fig:bounded-belt}.A). Call~$t_{\gamma}(x)$~the triangle
  formed by $u$,~$v$ and~$\gamma$. No $c$\dash/transversal can join
  $u$~and~$v$ inside~$t_{\gamma}(x)$. Otherwise $u$,~$v$ and this
  $c$\dash/transversal would form a triangle that would lift into a triangle
  in~$\unicov$, in contradiction with the triangle-free Property. Moreover, a
  $c$\dash/transversal~$w$ that crosses the $u$\dash/side of~$t_{\gamma}(x)$
  cannot be crossed inside~$t_{\gamma}(x)$ by any~$w'$, as $w$,~$w'$, $u$
  and~$v$ would then form a quadrilateral. If the $u$\dash/side
  of~$t_{\gamma}(x)$ is indeed crossed, we let $w_u$~be the crossing
  curve closer to~$x$ along the $u$\dash/side. We define $w_v$~similarly
  for the $v$\dash/side of~$t_{\gamma}(x)$. The $c$\dash/transversal curves
  $w_u$~and~$w_v$, if~any, together with $u$,~$v$ and~$\gamma$ bound a
  face of the subdivision of~$\mathcal A$. This is the only face incident
  to~$x$ in~$t_{\gamma}(x)$. We call it the \emph{left face} of~$x$;
  it has one side along~$\gamma$ and no side along~$\sigma$.
\begin{figure}
    \centering
    \includesvg[.85\linewidth]{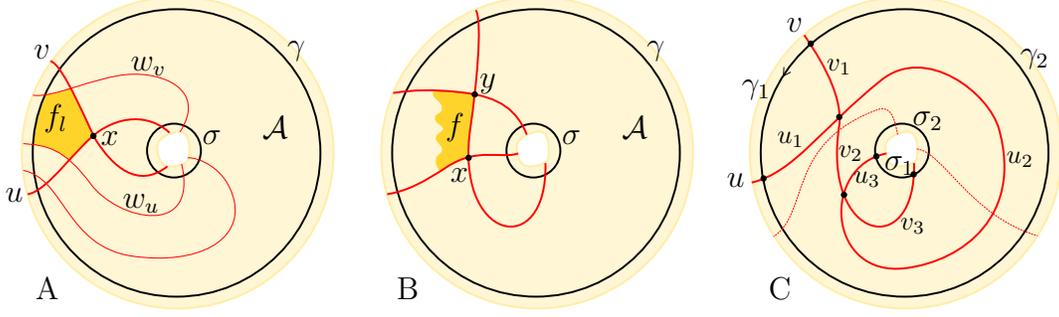}
    \caption{(A)~$f_l$~is the left face of~$x$ in the annulus~$\mathcal
      A$.\allowbreak\ (B)~The three supporting line of~$f$ passing through
      $x$~or~$y$.\allowbreak\ (C)~The lines $u$~and~$v$ cut~$\mathcal A$
      into simply connected faces.}
    \label{fig:bounded-belt}
  \end{figure}
  Conversely, we claim that every face~$f$ of the
  subdivision of~$\mathcal A$ with no side along~$\sigma$ is the left-face
  of a unique interior vertex. The unicity is clear since a left face
  has a unique side on~$\gamma$. In fact, by minimality, $\gamma$~crosses
  every face of~$\cyH^*$ at most once; so that $f$~has at most
  two (consecutive) vertices on~$\gamma$. If $f$~has a unique interior
  vertex on its boundary, then $f$~is precisely the triangle~$t_{\gamma}$
  for that vertex. Otherwise, consider two interior
  vertices $x$~and~$y$ that are consecutive along the boundary of~$f$
  as on Figure~\ref{fig:bounded-belt}.B. By
  considering the arrangement made by the three $c$\dash/transversals
  defining $x$~and~$y$, it is easily seen that $f$~is included in
  exactly one of $t_{\gamma}(x)$~or~$t_{\gamma}(y)$, thus proving the
  claim. We define \emph{right faces} analogously and remark that a face
  that is neither a left nor a right face must have one side along~$\gamma$
  and one side along~$\sigma$. 

  We can now determine the complexity of the subdivision of~$\mathcal
  A$. Denote by $V$~and~$E$ its respective numbers of vertices and
  edges. With the notations in the lemma, we have $V = V_I + 2|\gamma|$ and
  $E=E_I+2|\gamma|$. By the preceding remark, every face has a side on either
  $\gamma$~or~$\sigma$. It ensues that $F\leq 2|\gamma|$.  Euler's
  formula then implies $0=V-E+F= V_I - E_I + F$, whence $E_I\leq V_I +
  2|\gamma|$. Since interior and boundary vertices have respective
  degree $4$~and~$3$, we get $E_I = 2V_I + |\gamma|$ by
  double-counting of the vertex-edge incidences. Combining with the
  previous inequality we obtain $V_I\leq |\gamma|$, and finally
  conclude that $E_I\leq 3|\gamma|$.
  \paragraph{Case where at least two $c$\dash/transversals intersect
    twice.}
  We now suppose that two $c$\dash/transversals $u$~and~$v$ intersect twice
  in~$\mathcal A$. The curves $\gamma$,~$\sigma$, $u$ and~$v$ induce a
  subdivision of~$\mathcal A$ where $u$~and~$v$ are each cut into three
  pieces, say $u_1$,~$u_2$ and~$u_3$ for~$u$ and $v_1$,~$v_2$ and~$v_3$
  for~$v$ and the two generators are each cut into two pieces,
  say $\gamma_1, \gamma_2$ for~$\gamma$ and $\sigma_1,\sigma_2$ for~$\sigma$
  (see Figure~\ref{fig:bounded-belt}.C). A $c$\dash/transversal~$w$
  crossing~$\gamma_1$ would have to cross $u_1$~or~$v_1$ to enter
  the face bounded by $u_1\cdot v_2\cdot u_2\inv\cdot v_1\inv\cdot
  \gamma_2\inv$ in the subdivision induced by $\gamma$,~$u$ and~$v$.
  It is easily seen that no matter how this face is crossed, it
  will be cut into subfaces, one of which must have three or four
  sides bounded by $u$,~$v$ and~$w$. This contradicts
  Proposition~\ref{prop:lines:no-bigon-triangle}. Likewise no
  $c$\dash/transversal can cross~$\sigma_1$. Hence, any $c$\dash/transversal
  distinct from $u$~and~$v$ must extend between $\gamma_2$~and~$\sigma_2$
  and cut either $u_2$~or~$v_2$.  It is easily seen that
  any two such $c$\dash/transversals cannot cross without creating a
  triangle or a quadrilateral bounded by $c$\dash/transversals, which is
  again impossible. It follows that apart from $u$~and~$v$ all
  $c$\dash/transversals are pairwise disjoint inside~$\mathcal A$. We deduce
  that all faces have one side on $\gamma$~or~$\sigma$ (but not both),
  whence $F= 2|\gamma|$. Since any $c$\dash/transversal distinct from
  $u$~and~$v$ is cut into two pieces we also get $V_I=|\gamma|$ and
  $E_I=3|\gamma|$. We note that the left triangle $\gamma_1\cdot
  u_1\cdot v_1\inv$ is not cut by any $c$\dash/transversal.
\end{proof}
\paragraph{The \short/ edges.} An edge of~$\cyH^*$ whose
relative interior is crossed by a minimal generator is said
\define{\short/}.
\begin{lemma}\label{lem:short-connectivity}
  If~$\mathcal A$~is crossed by $c$\dash/transversal curves only, then every
  edge of~$\cyH^*$ in~$\mathcal A$ is \short/.
\end{lemma}
\begin{proof}
  If~$\mathcal A$~contains no vertex of~$\cyH^*$ then $\mathcal A$~crosses
  the same edges as $\gamma$~and~$\sigma$ and the lemma is trivial.
  Otherwise, we show how to sweep the entire arrangement inside~$\mathcal A$
  with a minimal generator from the left boundary~$\gamma$
  to the right boundary~$\sigma$. We claim that the subdivision of~$\mathcal
  A$ induced by $\gamma$,~$\sigma$ and~$\cyH^*$ contains a triangle
  face~$t$ with one side along~$\gamma$ (a left triangle in the above
  terminology). This is clear in the case where two $c$\dash/transversals
  cross twice as noted at the end of the above proof. Otherwise, we
  consider two $c$\dash/transversals crossing in~$\mathcal A$; they obviously
  form a left triangle with~$\gamma$. Adding the other
  $c$\dash/transversals one by one we see that this left triangle is
  either not cut or that a new left triangle is cut out of it. We conclude
  the claim by a simple induction on~$\len\gamma$. The left triangle~$t$
  has one vertex~$x$ interior to~$\mathcal A$ and incident to four edges
  $u_1,v_1,u_2,v_2$ where $u_1,v_1$ bound~$t$. We can now sweep~$x$
  with~$\gamma$ by crossing~$v_2,u_2$ instead of~$u_1,v_1$ to obtain a
  new minimal generator. This new generator bounds with~$\sigma$ a new
  annulus~${\mathcal A}'\subset{\mathcal A}$. Note that ${\mathcal A}'$~crosses
  the same set of edges as~$\mathcal A$ except for~$u_1, v_1$ that were
  crossed by~$\gamma$. Moreover, the number of interior vertices is
  one less in~${\mathcal A}'$ than in~${\mathcal A}$. We conclude the proof of
  the lemma with a simple recursion on this number.
\end{proof}
\begin{lemma}\label{lem:segment-short}
  The set of \short/ edges of a $c$\dash/transversal~$u$ is a finite
  segment of~$u$.
\end{lemma}
\begin{proof}
  Let $a$~and~$b$ be two \short/ edges along~$u$. Let $\gamma$~and~$\sigma$
  be minimal generators crossing $a$~and~$b$ respectively. By
  Lemma~\ref{lem:uncross-generators}, we can assume that $\gamma$~and~$\sigma$
  are disjoint and still cross $a$~and~$b$ (a single
  generator cannot cross a $c$\dash/transversal twice). If the annulus
  bounded by $\gamma$ and $\sigma$ is crossed by $c$\dash/transversals only
  then we can apply Lemma~\ref{lem:short-connectivity} to conclude that
  the edges between $a$~and~$b$ along~$u$ are \short/. Otherwise, we can
  cut this annulus into two parts as in the proof of
  Lemma~\ref{lem:short-bound} and reach the same conclusion. This
  latter lemma also implies that there are $O(|\gamma|)$~\short/ edges
  between $a$~and~$b$. It follows that set of \short/ edges along~$u$
  is connected and bounded.
\end{proof}
Consider two minimal generators $\mu$~and~$\nu$. We claim that there
exists a minimal generator that crosses the rightmost of the \short/
edges crossed by $\mu$~and~$\nu$ along each
$c$\dash/transversal. Indeed, the two disjoint minimal generators
$\gamma$~and~$\sigma$ returned by Lemma~\ref{lem:uncross-generators}
cannot invert their order of crossings along
$c$\dash/transversals. Hence one of them uses all the leftmost \short/
edges, while the other uses all the rightmost \short/ edges. By a
simple induction on the number~$|\gamma|$ of $c$\dash/transversals,
this implies in turn that there exists a minimal generator~$\gamma_R$
that crosses the rightmost \short/ edge of each $c$\dash/transversal.
We define the \define{canonical generator} with respect to~$c$ as the
cycle in~$\cyH$ dual to the sequence of \short/ edges crossed
by~$\gamma_R$.

\paragraph{The canonical belt}
As for~$\gamma_R$, we can show the existence of a minimal generator~$\gamma_L$
crossing the leftmost \short/ edges. We define the
\define{canonical belt}~${\mathcal B}_c$ as the union of the vertices,
edges and faces crossed by the annulus bounded by $\gamma_L$~and~$\gamma_R$.
By Lemma~\ref{lem:short-connectivity}, the edges in~${\mathcal
  B}_c$ are the \short/ edges and the edges of the projection of the
$\trans$\dash/invariant line, if any. In particular, any minimal generator
is included in the canonical belt.

\begin{lemma}
  Any intersection of $c$\dash/transversals is interior to the canonical
  belt.
\end{lemma}
\begin{proof}
  Consider two $c$\dash/transversals $u$~and~$v$ crossing at a vertex~$x$.
  Let~$\gamma$~be a minimal generator; it bounds, together with
  $u$~and~$v$, a triangle~$t$ in~$\cycov{c}$. Let~$\lambda$~be the
  generator obtained from~$\gamma$ by substitution of the $\gamma$\dash/side
  of~$t$ with the complementary part of the boundary of~$t$,
  slightly pushed out of~$t$. By
  the triangle-free Property, the projection of a
  line of~$\unicov$ that crosses the $u$\dash/side or $v$\dash/side of~$t$ must
  exit~$t$ through its $\gamma$\dash/side. It follows that $\lambda$~has the
  same \CW/ as~$\gamma$. Hence, $\lambda$~is minimal and the
  vertex~$x$, which is interior to the annulus bounded by
  $\gamma$~and~$\lambda$, is also interior to the canonical belt.
\end{proof}

We consider the subgraph~$\lift{K}^*$ of~$\lift{H}^*$ induced by the lines
in~$\lift{H}^*$ that are neither $\trans$\dash/transversal nor
$\trans$\dash/invariant.
The projection~$\cyproj(\lift{K}^*)$ of~$\lift{K}^*$ in~$\cycov{c}$ is
denoted by~$K_c^*$; it is the union of the projections of the lines that
are neither $\trans$\dash/transversal nor
$\trans$\dash/invariant. The following lemma gives a simple
characterisation of the canonical belt.
\begin{lemma}\label{lem:belt-characterisation}
  ${\mathcal B}_c$~is the unique component of~$\cycov{c}\setminus
  K_c^*$ that contains a generator.
\end{lemma}
\begin{proof}
  It is easily seen from its definition and from the previous lemma
  that ${\mathcal B}_c$~is a component of~$\cycov{c}\setminus K_c^*$.
  Let~$\sigma$~be any generator contained in~$\cycov{c}\setminus K_c^*$
  and let $\gamma$~be a (simple) minimal generator. Since $\gamma$~is
  included in~${\mathcal B}_c$, we just need to show that $\sigma$~and~$\gamma$
  are contained in the same component of~$\cycov{c}\setminus K_c^*$.
  If~$\sigma$~is not a simple curve, we
  first remove its contractible loops to make it simple. If
  $\sigma$~and~$\gamma$ intersect then there is nothing to show. Otherwise,
  $\sigma$~and~$\gamma$ bound a compact annulus~$\mathcal A$ in~$\cycov{c}$.
  As already observed, any line whose projection stays in~$\mathcal A$
  must be $\trans$\dash/invariant. It follows that $\mathcal A$~is
  entirely contained in a single component of~$\cycov{c}\setminus
  K_c^*$.
\end{proof}

\subsection{Computing the canonical generator}
We now explain how to compute the canonical generator associated with
the loop~$c$ in time proportional to~$|c|$.  Let~$\region{}$~be the
relevant region of the loop~$c^6$ obtained by six concatenations of~$c$.
According to Section~\ref{subsec:algo-contractibility}, we can
build the adjacency graph~$\Gamma$ of the faces of~$\region{}$ in
$O(|c|)$~time. The edges dual to the edges of~$\Gamma$ are the edges
of~$\uniH^*$ interior to~$\region{}$. They induce a subgraph of~$\uniH^*$
which we denote by~$\Gamma^*$. As opposed to~$\Gamma$, the
graph~$\Gamma^*$ may have multiple components (see
Figure~\ref{fig:poincare}). A vertex of~$\Gamma^*$ may have degree one
or four depending on whether it is a flat boundary vertex or an
interior vertex of~$\Pi$. As explained in the mirror sub-procedure of
Section~\ref{subsec:algo-contractibility}, if $\lift{e}^*\in\Gamma^*$~is
an edge dual to~$\lift{e}$ then the two siblings of~$\lift{e}$
correspond to the facing edges of~$\lift{e}^*$ and they can be
computed in constant time (if they belong to~$\Gamma^*$). Similarly,
we can easily compute in constant time the circular list of (one or
four) edges sharing a same vertex of~$\Gamma^*$.

\paragraph{Identifying and classifying lines in~$\Gamma^*$.}
From the preceding discussion we can traverse~$\Gamma^*$ to give a
distinct tag to each maximal component of facing edges in constant
time per edge. Lemma~\ref{lem:region:connected-inter-with-line}
ensures that each such component is supported by a distinct line. We
denote by~$\ell(\lift{e})$ the identifying tag of the line supporting
the edge dual to~$\lift{e}$. With a little abuse of notation we will
identify a line with its tag.

Let~$c_1,\ldots,c_6$~be the successive lifts of~$c$ in the lift of~$c^6$
and let $x_0, x_1,\ldots,x_6$~be the successive lifts of~$c(0)$
in the lift of~$c^6$.  Let~$\lift{e}_{i,j}$~be the
$j$\dash/th edge of~$c_i$.  Since by construction
$\trans(c_i) = c_{i+1}$, we have $\trans(\ell(\lift{e}_{i,j})) =
\ell(\lift{e}_{i+1,j})$. This allows us to compute the translate of any line
crossing one of $c_1,\ldots,c_5$ in constant time per line. Notice
that the interior of~$\region{}$ is crossed by a $\trans$\dash/invariant line if
and only~if $\trans(\ell)=\ell$ for some line~$\ell$ crossing one
(thus any) of~$c_1,\ldots,c_5$. We can now fill a table~$C[\ell]$
whose Boolean value is true if~$\ell$ intersects~$\trans(\ell)$ in~$\region{}$
and false otherwise. This clearly takes $O(|c|)$~time.  In order
to identify the $\trans$\dash/transversals separating $x_i$~from~$x_{i+1}$,
we need the following property.
\begin{lemma}\label{lem:relevant-crossing}
    Two intersecting lines crossing the interior of~$\region{}$ intersect
  in the interior of~$\region{}$.
\end{lemma}
\begin{proof}
  Otherwise there would be a line~$\lambda$ separating~$\region{}$ from the
  intersection of the two intersecting lines. These two lines would
  form a triangle with~$\lambda$, in contradiction with
  the triangle-free Property.
\end{proof}
We first identify and orient the $\trans$\dash/transversals separating
$x_2$~from~$x_{3}$. We start by filling a table~$P[\ell,i]$ counting
the parity of the number of intersections of each line~$\ell$ with~$c_i$
for~$i\in\{2,3,4\}$. This can be done in $O(|c|)$~time: we
initialize all the entries of the table~$P$ to~$0$ and, for~each~$i\in\{2,3,4\}$
and each edge~$\lift{e}$ of~$c_i$, we invert the
current parity of~$P[\ell(\lift{e}),i]$. By
Proposition~\ref{prop:lines:transversal-condition}, the
$\trans$\dash/transversals separating $x_2$~from~$x_{3}$ are exactly
those~$\ell$ for which $C[\ell]$~is false, $P[\ell,3]$~is odd, and
$P[\ell,2]$~and~$P[\ell,4]$ are even. We then orient these
transversals from left to right as follows. We traverse in order the
oriented edges of the lift~$c_3$. As we traverse an oriented edge~$\lift{e}$
of~$c_3$, we check whether its dual belongs to an identified transversal
that was not already oriented. In the affirmative, we use the
left-to-right orientation of the dual edge of~$\lift{e}$ --- given by the
correspondence between the oriented cellular embedding and its dual --- to
orient this transversal.

We can now identify and orient all the $\trans$\dash/transversals
separating $x_i$~from~$x_{i+1}$, for~$i\in[0,5]$, by translation of
those separating $x_2$~from~$x_{3}$. Corollary~\ref{cor:transversal-unic}
ensures the existence of at least one transversal separating
$x_2$~from~$x_{3}$; we choose one
and denote it by~$\ell$ in the sequel. We shall concentrate on the
part of~$\region{}$ contained in the closure~$\bar{B}_{\ell}$ of the
fundamental domain of~$\langle\trans\rangle$ comprised between
$\ell$~and~$\trans(\ell)$. We put ${\mathcal C} :=
\region{}\cap\bar{B}_{\ell}$.

\paragraph{Finding a lift of the canonical generator.}
We want to show that ${\mathcal C}$~contains either a whole lift of the
canonical belt or half of it. In this latter case, $\region{}$~is bounded by
a $\trans$\dash/invariant line.
\begin{lemma}\label{lem:two-situations}
  Exactly one of the two following situations occurs
  \begin{enumerate}
    \item $\region{}$~contains the intersection~$\bar{B}_{\ell}\cap
    \cyproj\inv({\mathcal B}_c)$.
  \item There exists a $\trans$\dash/invariant line~$\lambda$, whose
    projection~$\cyproj(\lambda)$ cuts the canonical belt into two
    open parts ${\mathcal B}_L$~and~${\mathcal B}_R$, each one
    containing a generator and intersecting exactly one edge of each
    $c$\dash/transversal. The relevant region~$\region{}$ contains
    either (i) the intersection $\bar{B}_{\ell}\cap
    \cyproj\inv({\mathcal B}_L)$ or (ii) the intersection
    $\bar{B}_{\ell}\cap \cyproj\inv({\mathcal B}_R)$ and exclude the
    other one.
  \end{enumerate}
\end{lemma}
\begin{proof}
  We first consider the case where there is no $\trans$\dash/invariant
  line. Let~$\lift{e}^*$ be an edge of~$\uniH^*\cap\bar{B}_{\ell}$
  projecting to an edge~$e^*$ of the canonical belt. We want to prove
  that $\lift{e}^*$~is interior to~$\region{}$. Denote by~$\ell'$ the
  supporting line of~$\lift{e}^*$; this is a $\trans$\dash/transversal by
  Proposition~\ref{prop:transversal-existence}. On the one hand, if
  $\ell'$~does not cut any other transversal then it stays
  inside~$\bar{B}_{\ell}$ and must separate $x_3$~from either $x_2$~or~$x_4$.
  In particular, $\ell'$~intersects the concatenation~$c_3\cdot
  c_4$, hence~$\region{}$. If one of the~$c_i$'s intersect~$\lift{e}^*$,
  then $\lift{e}^*$~is inside~$\region{}$ by definition. Otherwise, we
  consider a minimal generator~$\gamma$ crossing~$e^*$ and its
  reciprocal image~$\ell_{\gamma}$. Let~$y\in\ell'\cap (c_3\cdot c_4)$
  and let $z$~be the endpoint of~$\lift{e}^*$ between $y$~and the
  point~$\ell_{\gamma}\cap \lift{e}^*$. See
  Figure~\ref{fig:gamma_R-in-Pi}.
  \begin{figure}
    \centering
    \includesvg[.5\linewidth]{gamma_R-in-Pi}
    \caption{The grey region~$\mathcal C$ is the intersection of the
      relevant region~$\region{}$ of the lift of~$c^6$ with the fundamental
      domain bounded by $\ell$~and~$\trans(\ell)$.}
    \label{fig:gamma_R-in-Pi}
  \end{figure}
  The line~$\ell''$ crossing~$\ell'$ at~$z$ does not cut~$\ell_{\gamma}$
  by Proposition~\ref{prop:transversal-existence}.
  Lemma~\ref{lem:3-translates-invariant} implies that $\ell''$~does
  not cut $\trans^{-2}(\ell')$~nor~$\trans^3(\ell')$.
  It ensues that $\ell''$~intersects the lift of~$c^6$ between
  $\trans^{-2}(y)$~and~$\trans^2(y)$. Since
  $\ell'$~and~$\ell''$ both cross the interior of~$\region{}$, their
  intersection~$z$ is interior to~$\region{}$ by
  Lemma~\ref{lem:relevant-crossing}. It ensues that $\lift{e}^*$~is also
  interior to~$\region{}$. On the other hand, if~$\ell'$~indeed crosses some
  other transversal, then by Lemma~\ref{lem:segment-short} one of the
  endpoints of~$\lift{e}^*$ is the intersection of~$\ell'$ with some
  other $\trans$\dash/transversal~$\mu$. Lemma~\ref{lem:3-translates-invariant}
  implies that these two transversals cannot intersect
  $\trans^{-2}(\ell)$~nor~$\trans^3(\ell)$. It ensues that $\ell'$~and~$\mu$
  cross the lift of~$c^6$  (this is where we need six lifts of~$c$)
  and we conclude as above that
  $\lift{e}^*$~is interior to~$\region{}$.

  We now consider the case where there exists a $\trans$\dash/invariant
  line~$\lambda$. As observed in the proof of
  Lemma~\ref{lem:1-invariant-line}, any $\trans$\dash/transversal must
  cross~$\lambda$. The triangle-free Property then
  forbids these transversals to intersect. It follows that every
  $\trans$\dash/transversal contains exactly two short edges separated by a
  vertex of~$\lambda$. Said differently, the edges of the canonical
  belt precisely lift to
  the edges incident to~$\lambda$ (including the edges of~$\lambda$). On the
  one hand, if $\lambda$~does not cross any of the~$c_i$'s, then the
  lift of~$c^6$ lies on one side of~$\lambda$, say its left. The line~$\lambda$
  may play the role of~$\ell_{\gamma}$ and we may argue as
  above that all short edges to the left of~$\lambda$ and contained
  in~$\bar{B}_{\ell}$ are interior to~$\region{}$. Equivalently, every lift
  in~$\bar{B}_{\ell}$ of an edge of~${\mathcal B}_L$ is interior to~$\region{}$.
  On the other hand, if $c$~crosses~$\cyproj(\lambda)$, then $\lambda$~crosses
  the interior of~$\region{}$. Lemma~\ref{lem:relevant-crossing}
  ensures that all its intersections with the $\trans$\dash/transversals
  in~$\bar{B}_{\ell}$ are interior to~$\region{}$. We conclude this time that
  $\region{}$~contains $\bar{B}_{\ell}\cap \cyproj\inv({\mathcal B}_c)$.
\end{proof}
We now explain how to identify the lift of the canonical belt, or half
of it, contained in~$\mathcal C$.
\begin{lemma}\label{lem:identify-sigma-star}
  Let~$\Sigma^*$~be the subgraph of~$\Gamma^*\cap{\mathcal C}$ projecting
  to the canonical belt.
  We can identify the edges of~$\Sigma^*$ in $O(|c|)$~time.
\end{lemma}
\begin{proof}
  From the preceding lemma, $\cyproj({\mathcal C})\cap {\mathcal B}_c$~is
  connected and contains a generator. Recall that $\lift{K}^*$~is the
  union of the lines that are neither $\trans$\dash/transversal nor
  $\trans$\dash/invariant, and that $K_c^*$~is its projection into~$\cycov{c}$.
  Lemma~\ref{lem:belt-characterisation} ensures that
  $\cyproj({\mathcal C})\cap {\mathcal B}_c$~is the only component
  of~$\cyproj({\mathcal C})\setminus K_c^*$ that contains a
  generator. Equivalently, ${\mathcal C}\cap \cyproj\inv({\mathcal B}_c)$~is
  the only component of~${\mathcal C}\setminus \lift{K}^*$ that contains
  an edge~$\lift{e}^*\in\ell$ together with its translate~$\trans(\lift{e}^*)$.

  Thanks to Lemma~\ref{lem:3-translates-invariant} and following the
  paragraph on line identification and classification, we can
  correctly detect all the $\trans$\dash/transversal and $\trans$\dash/invariant
  lines crossing~$\mathcal C$. By complementarity, we identify the edges
  of~$\lift{K}^*$ in~${\mathcal C}$. We also identify by a simple
  traversal the subgraph~$\Gamma_{\mathcal C}$ of~$\Gamma$ whose dual
  edges are contained in~$\mathcal C$. The graph~$\Gamma_{\mathcal C}$ is
  connected ($\mathcal C$~is a ``convex'' region of the plane) and each
  component of~${\mathcal C}\setminus \lift{K}^*$ is a union of faces
  corresponding to a connected component of~$\Gamma_{\mathcal C}\setminus
  \lift{K}$, where $\lift{K}$~is the set of primal edges corresponding
  to~$\lift{K}^*$. We eventually select the component~$\Sigma$
  of~$\Gamma_{\mathcal C}\setminus K$ that includes an edge~$\lift{e}$
  together with its translate~$\trans(\lift{e})$. It clearly takes time
  proportional to~$|c|$ to select the edges of~$\Sigma$. We finally remark
  from the initial discussion that the dual of the edges in~$\Sigma$ are the
  edges of~$\Sigma^*$.
\end{proof}
\begin{proposition}\label{prop:canonical-generator}
  We can compute the canonical generator in~$O(|c|)$ time.
\end{proposition}
\begin{proof}
  We first compute~$\Sigma^*$ as in
  Lemma~\ref{lem:identify-sigma-star}. We then determine if we are in
  situation~2\textit{(i)} of Lemma~\ref{lem:two-situations}. To this end, we
  check that all the edges of~$\Sigma^*$ are supported by pairwise
  distinct $\trans$\dash/transversals. If this is the case, we further
  check if the right extremities of the edges of~$\Sigma^*$ are linked
  by edges of a (necessarily $\trans$\dash/invariant) line. This is easily
  seen in constant time per edge by projecting the edges of~$\Sigma^*$
  back into~$H^*$ on~$\surf$. If we are indeed in situation~2\textit{(i)},
  the canonical generator is composed of the projection on~$\cycov{c}$
  of the dual of the edges facing the edges of~$\Sigma^*$ to their
  right (keeping only one of the two edges supported by
  $\ell$~and~$\trans(\ell)$). In the other situations 1~and~2\textit{(ii)},
  the edges crossed by the lift of the canonical generator in~$\bar{B}_{\ell}$
  are the edges of~$\Sigma^*$ that are supported by
  $\trans$\dash/transversal lines and whose right endpoint is not a
  crossing with any other $\trans$\dash/transversal or $\trans$\dash/invariant
  line. In other words, these are the rightmost edges in~$\Sigma^*$ of
  the pieces of $\trans$\dash/transversals crossing~$\Sigma^*$, unless they
  abut on $\ell$~or~$\trans(\ell)$. In either case, we can clearly
  determine the sequence of edges of the canonical generator in
  $O(|c|)$~time.
\end{proof}
\subsection{End of the proof of Theorem~\ref{th:free-homotopy-test} }
Let $c$~and~$d$ be two non-contractible cycles represented as closed
walks in~$H$. Assuming that $\surf$~is orientable with genus at least
two, we compute the canonical generators $\gamma_R$~and~$\delta_R$
corresponding to $c$~and~$d$ respectively. This takes $O(|c|+|d|)$~time
according to
Proposition~\ref{prop:canonical-generator}. Following the discussion
in the Introduction, $c$~and~$d$ are freely homotopic if and only if
the projections $\pi_c(\gamma_R)$~and~$\pi_c(\delta_R)$ in~$\surf$ are
equal as cycles of~$H$. This can be determined, under the obvious
constraint that these two projections have the same length, in
$O(|c|+|d|)$~time using the classical Knuth-Morris-Pratt
algorithm~\cite{clrs-ia-09} to check whether $\pi_c(\gamma_R)$~is a
substring of the concatenation~$\pi_c(\delta_R)\cdot~\pi_c(\delta_R)$.

It remains to consider the cases of $\surf$~being
a torus. As noted for the contractibility test, its fundamental
group is commutative. Being conjugate is thus equivalent to being equal as
group elements and the free homotopy test reduces to the contractibility
test.

We have thus solved the free homotopy test for closed orientable
surfaces. In a forthcoming paper, we shall show how to solve the
contractibility and free homotopy tests in optimal linear time for
surfaces with non-empty boundary, orientable or not. (In the orientable
case, we could also resort to the present algorithms by first closing
each boundary with a punctured torus.) We leave the
handling of the free homotopy test on closed non-orientable surfaces
as an open problem. There are two easy cases, though. If $\surf$~is a
projective plane, its fundamental group is again commutative and the
test is trivial. If $\surf$~is a Klein bottle, the test was actually
resolved by Max Dehn (see its papers translation by
Stillwell~\cite[p.153]{s-pgtt-87}): following the end of
Section~\ref{sec:contractibility-test}, we can assume that we are
given two words in their canonical forms $a^ub^v$~and~$a^{u'}b^{v'}$.
These two words are conjugate if and only~if $v=v'$
and either $v$~is even~and $u=u'$, or $v$~is odd~and $u$~has the same
parity as~$u'$. For non-orientable surfaces of larger genus we finally
note that if $c$~and~$d$ are two-sided curves in~$\surf$, which is
easily checked in linear time, their cyclic coverings are cylinders
and we can apply our algorithm for the orientable case. The only
unsolved case is when $c$~and~$d$ are one-sided curves and their
cyclic coverings are M\"obius bands. We can still define their
canonical belt as for two-sided curves. It has a single boundary
freely homotopic to the square of a generator. We can easily
checked that the canonical belts are isomorphic in time that is
quadratic in their size. It is not clear how to reduce this comparison
to linear time.

\appendix
\section*{Appendix: Counter-examples to Dey and Guha's approach}
\label{sec:counterexamples}
In a first stage, Dey and Guha~\cite{dg-tcs-99} obtain a term product
representation of $c$ and $d$ as in the present
Lemma~\ref{lem:term-product}. Suppose $f_H|_A = a_1a_2\cdots a_{4g}$,
then a term~$a_i a_{i+1}\cdots a_j$ is denoted~$(i,j)$.  This term is
equivalent in~$\langle A\,;\, f_H|_A\rangle$ to the complementary
term~$a_{i-1}\inv a_{i-2}\inv\cdots a_{j+1}\inv$ \emph{going backward}
along~$f_H|_A$. This complementary term is denoted~$\overline{(i-1,j+1)}$.
The \define{length}~$|(i,j)|$ of a term~$(i,j)$ is the length of the
sequence~$a_i a_{i+1}\cdots a_j$.  The
length of a complementary term is defined analogously, so that
$|(i,j)| + |\overline{(i-1,j+1)}| = 4g$. The \define{length} of a
product of (possibly complementary) terms is the sum of the lengths of
its terms. Let us rename the above term and complementary term as
respectively a \define{forward term} and a \define{backward term}. A
term will now designate either a forward or backward term. Note that a
term being equivalent to its complementary term, we may use a forward
or backward term in place of each term. By convention, we will write a
term in backward form only if it is strictly shorter than its
complementary forward term. This convention will be implicitly assumed
in this section and corresponds to the notion of \emph{rectified} term
in~\cite{dg-tcs-99}.

Let us say that a product~$t_1 t_2$ of two terms
\begin{itemize}
\item $0$\dash/\define{reacts} if~$t_1 t_2 = 1$, the unit element in the
  group~$\langle A\,;\, f_H|_A\rangle$,
\item $1$\dash/\define{reacts} if~$t_1 t_2 = t$ in~$\langle A\,;\,
  f_H|_A\rangle$, for a term~$t$ such that $|t| \leq |t_1| + |t_2|$,
  and
\item $2$\dash/\define{reacts} if~$t_1 t_2 = t'_1 t'_2$ in~$\langle A\,;\,
  f_H|_A\rangle$, for two terms~$t'_1, t'_2$ such that $|t'_1| +
  |t'_2| < |t_1| + |t_2|$.
\end{itemize}
The aim of Dey and Guha is to apply reactions to a given term product
in order to reach a canonical form where no two consecutive terms
react in that form. For this, they define a function  $\texttt{apply}$
that recursively applies reductions to a product of terms. This
function is in turn called by another function $\texttt{canonical}$,
supposed to produce a canonical form.

The following claim appears as points 2~and~3 in Lemma~4
of~\cite{dg-tcs-99} and aims at showing that the function~$\texttt{apply}$ does
terminate.

\nobreak
\emph{
Let~$u,v,w$~be $3$~terms such that $u v$~does not react. If~$v w$
$1$\dash/reacts~or $2$\dash/reacts with $v w = v'$~or $v w = v' w'$
(and $v' w'$~does not react), then $uv'$~does not $1$\dash/react.
}

The non-existence of such $1$\dash/reactions is essential in the proof that
the function $\texttt{canonical}$ indeed returns a canonical
form~\cite[Prop. 7]{dg-tcs-99}. However,
this claim is false as demonstrated by the following examples.
Consider a genus~$2$ surface with $f_H|_A = abcda\inv b\inv c\inv
d\inv$. Put $u=(2,4)$,~$v=\overline{(1,7)}$, and~$w = (7,8)$. Then
$uv = bcd\cdot a\inv dc$ does not react and $vw = a\inv dc\cdot c\inv
d\inv$ $1$\dash/reacts, yielding~$v' = a\inv$. But $uv'$~$1$\dash/reacts, in
contradiction with the claim, since $uv'=bcd\cdot a\inv = (2,5)$.
Likewise, if we now set $u = (2,4)$,~$v = \overline{(1,8)}$ and~$w =
\overline{(4,2)}$, we have: $uv$~does not react, $vw$~$2$\dash/reacts,
yielding~$v' w' = (5,5)\cdot\overline{(3,2)}$, and $uv'$~$1$\dash/reacts, in
contradiction with the claim, since $uv'=bcd\cdot a\inv = (2,5)$.

Define the \define{expanded word} of a term product as the word in the elements
of~$A$ (and their inverses) obtained by replacing each term in the
product with the corresponding factor of~$f_H|_A$ or~$(f_H|_A)\inv$.
Again, $f_H|_A$~and~$(f_H|_A)\inv$ should be
considered cyclically. Call a product of terms \define{stable} if no two
  consecutive terms react.
Another important claim~\cite[Lem. 8]{dg-tcs-99} states that

\nobreak
\emph{ The expanded word of a stable product of terms does not contain
  a factor of length~$2g+1$ that is also a factor of
  $f_H|_A$~or~$(f_H|_A)\inv$.
}

This claim is used to prove that the (supposed) canonical form of a
product is equivalent to~$1$ if and only if it is the empty
product~\cite[Prop. 6]{dg-tcs-99}. However this claim is again false
as demonstrated by the following example. Consider the same genus~$2$
surface as in the previous examples. Then the product
$\overline{(1,7)}\cdot (2,4)\cdot \overline{(1,7)} = cba\cdot bcd\cdot
a\inv dc$ is stable and contains the factor~$a\cdot bcda\inv$ of
length~$2g+1=5$ that is also a factor of~$(f_H|_A)$.

Finally, the canonical form defined by Dey and Guha is not canonical. By
definition of the $\texttt{canonical}$ function in~\cite[p.
314]{dg-tcs-99}, a stable (rectified) product~$w$ is canonical,
\ie/ $\texttt{canonical}(w)~=~w$. Using the same genus~$2$ surface as
before, consider the products $w_1 = \overline{(8,6)}\cdot\overline{(8,6)}
= dcb\cdot dcb$ and $w_2 = (1,4)\cdot (2,5) = abcd\cdot bcda\inv$. It
is easily seen that none of these products react. It follows that \texttt{
  canonical}$(w_i)~=~w_i$, $i =1,2$. However $w_1 = w_2$ in~$\langle A\,;\,
  f_H|_A\rangle$. Indeed, since $\overline{(8,6)} = (1,5)$ in~$\langle A\,;\,
  f_H|_A\rangle$, we have
\[ w_1 = abcda\inv\cdot abcda\inv = abcd\cdot bcda\inv = w_2
\] 
This contradicts the fact that an element of~$\langle A\,;\,
f_H|_A\rangle$ can be expressed as a unique canonical product of
terms. In particular, Proposition~7 in~\cite{dg-tcs-99} is wrong.

The method proposed in~\cite{dg-tcs-99} seems bound to fail: the above
counterexamples\footnote{The counterexamples easily generalize to
  genus $g > 2$ orientable surfaces with $f_H|_A = a_1a_2\cdots
  a_{2g}a_1\inv a_2\inv\cdots a_{2g}\inv$. Similar counterexamples for
  non-orientable surfaces can also be found starting with the product
  of squares as a (canonical) relator.} show that there are not enough
rules to ensure that stable products are canonical forms, but adding
more rules would make more difficult to keep under check chains of
reactions which are already out of control. Another issue is that for
the end of the comparison algorithm to be usable, appending a term to
a stable product should not trigger chains of reaction deep down the
stack. This adds even more constraints on the reaction rules,
especially as lemma 13 in~\cite{dg-tcs-99} does not seem to
help. Indeed, the function $\texttt{good\_conjugate}$ introduced in
the lemma is supposed to reinforce the function
$\texttt{reduced\_conjugate}$ defined in~\cite[p.~318]{dg-tcs-99}.  
From its description, $\texttt{good\_conjugate}$ transforms a
canonical product $w$ in $\tilde{w}$ by inserting some term
product~$c\cdot c\inv$ into~$w$ before computing $\tilde{\tilde{w}} =
\texttt{canonical}(\tilde{w})$ and
$\texttt{reduced\_conjugate}(\tilde{\tilde{w}})$. If~$\texttt{canonical}$~were
canonical then $\tilde{\tilde{w}}$~and~$w$ would be equal as products
of terms; it follows that the functions $\texttt{good\_conjugate}$ and
$\texttt{reduced\_conjugate}$ would have exactly the same effect.

\end{document}